\documentclass[11pt]{article}  
\usepackage{amsmath,amssymb}

\usepackage{amsthm}
\usepackage[margin=1in]{geometry}
\usepackage{graphicx}

\usepackage{algorithm}
\usepackage[noend]{algpseudocode}

\DeclareMathOperator{\polylog}{polylog}
\DeclareMathOperator{\poly}{poly}

\begin{document}

\newcommand{\bE}{\ensuremath{\mathbf{E}}}
\newtheorem{theorem}{Theorem}[section]
\newtheorem{proposition}[theorem]{Proposition}
\newtheorem{lemma}[theorem]{Lemma}
\newtheorem{corollary}[theorem]{Corollary}
\newtheorem{definition}[theorem]{Definition}

\newtheorem*{rep@theorem}{\rep@title}
\newcommand{\newreptheorem}[2]{%
\newenvironment{rep#1}[1]{%
 \def\rep@title{#2 \ref{##1}}%
 \begin{rep@theorem}}%
 {\end{rep@theorem}}}
\makeatother

\newreptheorem{proposition}{Proposition}

\setcounter{page}{0}
\title{Derandomized concentration bounds for polynomials, and hypergraph maximal independent set}
\author{David G. Harris\thanks{Department of Computer Science, University of Maryland, 
College Park, MD 20742.
Email: \texttt{davidgharris29@gmail.com}}}

\maketitle

\begin{abstract}
  A parallel algorithm for maximal independent set (MIS) in hypergraphs has been a long-standing algorithmic challenge, dating back nearly 30 years to a survey of Karp \& Ramachandran (1990). The best randomized parallel algorithm for hypergraphs of fixed rank $r$ was developed by Beame \& Luby (1990) and Kelsen (1992), running in time roughly $(\log n)^{r!}$.

  We improve the randomized algorithm of Kelsen, reducing the runtime to roughly $(\log n)^{2^r}$ and simplifying the analysis through the use of more-modern concentration inequalities. We also give a method for derandomizing concentration bounds for low-degree polynomials, which are the key technical tool used to analyze that algorithm. This leads to a deterministic PRAM algorithm also running in $(\log n)^{2^{r+3}}$ time and $\poly(m,n)$ processors. This is the first deterministic algorithm with sub-polynomial runtime for hypergraphs of rank $r > 3$.

  Our analysis can also apply when $r$ is slowly growing; using this in conjunction with a strategy of Bercea et al. (2015) gives a deterministic MIS algorithm running in time $\exp(O( \frac{\log (mn)}{\log \log (mn)}))$.
\end{abstract}

This is an extended version of a paper appearing in the ACM-SIAM Symposium on Discrete Algorithms (SODA) 2018.

\pagebreak

\section{Introduction}
Let $G = (V,E)$ be a hypergraph of rank $r$ on $n$ vertices and $m$ hyper-edges (a rank-$r$ hypergraph means that every edge has cardinality at most $r$). An \emph{independent set} of $G$ is a subset $I \subseteq V$ such that $e \not \subseteq I$ for all edges $e \in E$; a maximal independent set (MIS) is a set $I$ which is independent, but $I \cup \{v \}$ is not independent for $v \in V \setminus I$. 

It is trivial to find an MIS by a sequential algorithm. For ordinary graphs (with $r = 2$), finding an MIS is a fundamental symmetry-breaking problem in distributed/parallel computing. This problem has a long history, with efficient randomized parallel (RNC) and deterministic parallel (NC) algorithms dating back to Luby \cite{luby}.  Hypergraph MIS, by contrast,  has been a long-standing open challenge problem, going back nearly 30 years to the survey of Karp \& Ramachandran \cite{kr}. Despite its superficial simplicity, and the fact that it has a trivial sequential algorithm, and its similarity to the well-understood graph MIS problem, there have been no general parallel algorithms (NC or RNC). On the other hand there have been no hardness results for this problem either.

In \cite{karp}, Karp et al. described a randomized algorithm with runtime roughly $\sqrt{n}$; this remains the best time complexity of any general randomized algorithm. No deterministic algorithms are known for the general case in $o(n)$ time. A variety of special cases have efficient algorithms: in \cite{luczak}, \L{}uczak \& Szymanska gave an RNC algorithm for linear hypergraphs, in \cite{syoudai}, Syoudai \& Miyano gave an RNC algorithm for hypergraphs with bounded vertex-degree, in \cite{garrido}, Garrido, Kelsen, \& Lingas gave an NC algorithm for hypergraphs of bounded arboricity, and in \cite{dahlhaus}, Dahlaus et al. gave an NC algorithm for hypergraphs of maximum rank $3$.

More relevant for our paper, in \cite{beame-luby} Beame \& Luby gave an RNC algorithm for hypergraphs with maximum rank $r = 3$; this was subsequently extended by Kelsen \cite{kelsen} to cover any fixed value of $r$, using $(\log n)^{c_r}$ time and $\text{poly}(m,n)$ processors for $c_r$ (roughly) of order $r!$. In \cite{bercea}, Bercea et al. gave a more precise analysis showing a runtime of $(\log n)^{\tfrac{1}{2} (r+1)! + O(1)}$ (the constant term is independent of $r$). Bercea et al. used this as a subroutine to get a general randomized algorithm in $\exp( O( \frac{\log m \log \log \log m}{\log \log m}))$ time.  Along similar lines, Kutten et al. \cite{kutten} adapted Kelsen's algorithm to obtain distributed algorithms with approximately the same runtime as the parallel algorithm.

\textbf{Concentration bounds for polynomials and hypergraph MIS.} One strategy for analyzing graph MIS algorithms is to show that the graph degree decreases exponentially. This becomes more complicated for hypergraph MIS algorithms, as the hypergraph degree (suitably defined) is not even guaranteed to decrease monotonically. The main breakthrough of Kelsen \cite{kelsen} was to show that, with high probability, the degree increases slowly in every round. This is almost as good as guaranteeing that the degree decreases.

It is straightforward to show that the \emph{expected} increase in degree is small during the hypergraph MIS algorithm. To show a high-probability bound, Kelsen thus developed a new technical tool: an exponentially-strong concentration bound for polynomials applied to independent random variables. This generalizes the setting of concentration bounds for sums of independent random variables. The bound used in that paper, while adequate for the algorithm analysis, was somewhat crude. Since then, the work of Kelsen has given rise to an extensive line of research showing tighter and more general bounds \cite{vu1, vu2, kim-vu, schudy-sviridenko}. These in turn have found numerous applications in  combinatorics and computer science.
 
\textbf{Deterministic algorithms.} Obtaining deterministic parallel MIS algorithms has proven much more challenging. The algorithm of Dahlaus et al. \cite{dahlhaus} was specialized to $r=3$ and appears impossible to generalize to arbitrary values of $r$. There are no known deterministic algorithms running in polylogarithmic (or even sub-polynomial) time using polynomial processors, even for the case of constant $r > 3$.

In \cite{kelsen}, Kelsen discussed one straightforward approach to derandomize the randomized algorithm by drawing the random bits from a probability space with $k$-wise independence. When $k$ is constant, this space has polynomial size, and so can be searched efficiently. However, this leads to relatively weak concentration bounds, and so the resulting algorithm runs in $n^{\epsilon}$ time and $\poly(n)$ processors for any constant $r$ and $\epsilon > 0$. When $k \geq \Omega(\frac{\log n}{\log \log n})$, the algorithm runs in $\polylog(n)$ time, but requires super-polynomial processor count.

The probabilistic methods underpinning polynomial concentration bounds are similar to those for sums of independent random variables. There have been numerous powerful techniques developed for the latter, much more powerful than $k$-wise-independence for constant $k$. Unfortunately, there are severe technical roadblocks to applying these to non-linear polynomials. To the best of our knowledge, polynomial concentration bounds have not led to any efficient deterministic algorithms.

\subsection{Our contributions}
In Section~\ref{rand-sec}, we give a slightly modified form of Kelsen's randomized algorithm, and show an improved bound on its running time.
\begin{theorem}
\label{Fmain-thm}
There is a randomized parallel algorithm, running in $(\log n)^{2^{r+3}+O(1)}$ expected time and $O(n + m \log n)$ processors, to produce an MIS of a rank-$r$ hypergraph.
\end{theorem}

There are two main ingredients to this improvement. First, we use a concentration inequality due to Schudy \& Sviridenko \cite{schudy-sviridenko}, which is much tighter than the bounds originally developed by Kelsen. Second, we use an alternate degree statistic to measure the algorithm's progress. This measure is defined in terms of a single scalar value, which is used globally to bound the degrees throughout the graph. In addition to better running time, this substantially simplifies the analysis of \cite{beame-luby, bercea, kelsen}, which used multiple, interlocking potential functions.

Our second major contribution is to derandomize this algorithm:
\begin{theorem}
\label{Fdet-thm-overall}
There is a deterministic algorithm, running in $(\log n)^{2^{r+3} + O(1)}$ time and $\poly(m,n)$ processors, to produce an MIS of a rank-$r$ hypergraph.
\end{theorem}

This gives the first NC algorithm for arbitrary fixed $r$. The exponent in the running time matches that of the randomized algorithm while the exponent for the processor count is truly constant, not depending on $r$.

The technical core of Theorem~\ref{Fdet-thm-overall} (discussed in Section~\ref{Fedge-migrate-sec-det}) is a derandomization of polynomial concentration bounds. This is based on a potential function which serves as a pessimistic estimator for the bad event that the polynomial deviates significantly from its mean.  To summarize briefly, consider a polynomial $S(x_1, \dots, x_n)$ and independent Bernoulli random variables $X_1, \dots, X_n$, and let $Y = S(X_1, \dots, X_n)$. Many concentration bounds (for polynomials or for sums of independent random variables) are based on applying Markov's inequality to the random variable $Y^w$  for some suitably large (polylogarithmic) value of $w$. Thus, a natural potential function to bound the random variable $Y$ would be $\Phi(X_1, \dots, X_n) = \bE[Y^w]$ --- here, $\bE[Y^w]$ represents the \emph{conditional} expected value of $Y^w$, given that certain input values have been fixed.

When $S$ is non-linear, it appears to be impossible to efficiently compute $\bE[Y^w]$, so it cannot be used directly as a potential function. We use an alternative potential function in this situation; although the precise formula is complex, we can sketch the intuition behind it. The quantity $Y^w = S(X_1, \dots, X_n)^w$ can be regarded as a polynomial, in which the highest-degree terms correspond to sets of $w$ \emph{disjoint} monomials of $S$. In the cases of interest to us, these dominate the lower-degree terms. We therefore use as our potential function
$$
\Phi(X_1, \dots, X_n) = (\bE[Y] + \hat \delta)^w;
$$
where $\hat \delta$ is a crude upper bound on the lower-degree monomials, which, crucially, does not depend on the values of $X_1, \dots, X_n$. This potential function $\Phi$ \emph{can} be computed efficiently.

This approximation is only accurate when the partial derivatives of $S$ are relatively small. Thus, in parallel to maintaining concentration bounds for $S$, we will maintain concentration bounds for all its derivatives. These partial derivatives are themselves polynomials applied to independent input variables, so this procedure must be repeated recursively. Each stage of the recursion incurs a small error term. Keeping track of all these concentrations (in a mutually dependent way) is the most technically challenging part of our derandomization.

In Section~\ref{Fsparse-sec}, we leverage Theorems~\ref{Fmain-thm} and \ref{Fdet-thm-overall} to give new parallel algorithms for sparse graphs. This improves on the algorithm of \cite{bercea} in two ways: it is deterministic, and it is faster.
\begin{theorem}
\label{Fth-sbl1}
There is a randomized algorithm time for hypergraph MIS using  $\exp(O( \frac{\log (mn)}{\log \log (mn)}))$ expected time and $O(n + m \log n)$ processors. There is a deterministic algorithm for hypergraph MIS in $\exp(O( \frac{\log (mn)}{\log \log(mn)}))$ time and $\poly(m,n)$ processors.
\end{theorem}

Finally, in Section~\ref{conclude-sec} we conclude with discussion of limitations and further directions for the analysis of hypergraph MIS and derandomized polynomial concentration bounds. We hope that our derandomization method can lead to further deterministic algorithms via polynomial concentration.

\subsection{Notation}
We let $m$ denote the number of edges and $n$ the number of vertices in $G$. We assume throughout that $n$ is larger than any specified constants; in particular, we use without further comment any inequalities which only hold for sufficiently large $n$.

We use the Iverson notation, where $[\mathcal P]$ is the indicator function that any boolean predicate $\mathcal P$ holds. We use $\log x$ to denote the natural logarithm of $x$, and $\exp(x)$ to denote the exponentiation with base $e = 2.718...$.

We say an event $\mathcal E$ occurs \emph{with very low probability} (abbreviated \emph{wvlp}) if $\Pr(\mathcal E) < e^{-\Omega((\log m)^{1.01})}$, and we say that $\mathcal E$ occurs \emph{with very high probability (wvhp)} if $\neg \mathcal E$ occurs wvlp.

Given any collection of sets $\mathcal U$, we define the \emph{core} of $\mathcal U$ to be the inclusion-wise minimal sets of $\mathcal U$, i.e. $\text{core}(\mathcal U) = \{ U \in \mathcal U \mid U \not \subset W  \text{ for all $W \in \mathcal U \setminus \{ U \}$} \}$. For any hypergraph $G = (V,E)$ and independent set $I$ for $G$, we define the \emph{residual hypergraph of $G$ with respect to $I$}, which we write $G^{\text{res}}_I$, to be the hypergraph on vertex set $V$ with edge set $E' = \text{core} \bigl(  \{ e \setminus I \mid e \in E  \} \bigr)$. To put it less formally, we get $G^{\text{res}}_I$ from $G$ by removing the vertices of $I$, and then removing any edge which is a strict superset of another edge.

For any independent set $I$ of $G$, a maximal independent set $I' \supseteq I$ corresponds to a maximal independent set of $G'$. The restriction to the core is worth special explanation: if $G'$ contains nested edges $e \subsetneq e'$, then
 $e'$ is redundant: as long as the independent set does not contain $e$, it will automatically not contain $e'$. Thus $e'$ can be discarded without changing the maximal independent sets. This apparently inconsequential step turns out to be crucial to the algorithm.
\section{The randomized algorithm}
\label{rand-sec}
Our algorithm is based on a randomized procedure REDUCE developed by Beame \& Luby \cite{beame-luby} for successively building up an independent set $I$. It takes as input an additional parameter $p \in [0,1]$.

\begin{algorithm}[H]
\centering
\begin{algorithmic}[1]
\State Initialize $I = \emptyset$
\For{$t  = 1, \dots, T = (\log m)^{2^{r+2}}$}
\State Let $G^{(t)} = G^{\text{res}}_I$.
\State{For each vertex $u$, draw an independent Bernoulli-$p$ random variable $C^{(t)}(u)$.}
\State Form the set $A^{(t)} \subseteq V$, wherein $u \in A^{(t)}$ iff $C^{(t)}(u) = 1$ and for every edge $e  \in G^{(t)}$ with $u \in e$, there is some vertex $u' \in e$ with $C^{(t)}(u') = 0$.
\State{Update $I \leftarrow I \cup A^{(t)}$.}
\EndFor
\State Set $G^{(T+1)} = G_I^{\text{res}}$ and output $I$
\end{algorithmic}
\caption{REDUCE($G, p$)}
\end{algorithm}

Less formally, in each round $t$ we first randomly ``mark'' the vertices with probability $p$, and then perform an ``alteration'' step, in which every fully-marked edge of the residual hypergraph unmarks all its vertices.

Bercea et al. \cite{bercea} showed how to implement each individual round of REDUCE  in $\polylog(m,n)$ time and $O(n + m \log n)$ processors, so overall this algorithm uses $(\log m)^{2^{r+2} + O(1)} \polylog(n)$ time and $O(n + m \log n)$ processors.

In analyzing  $\text{REDUCE}(G, p)$, we use the superscript $(t)$ throughout to refer to properties of the hypergraph $G^{(t)}$. We also use the following shorthand: for any vertex set $X \subseteq V$, we define $C^{(t)}(X) = \prod_{x \in X} C^{(t)}(x)$. Thus, a simpler way to write step (5) is that $u \in A^{(t)}$ iff $C^{(t)}(u) = 1$ and $C^{(t)}(e) = 0$ for every each edge $e \in G^{(t)}$ which has $u \in e$.

\subsection{Measuring hypergraph degree}
As in \cite{beame-luby, kelsen}, we will define and track a measure of the hypergraph degree through the multiple rounds of REDUCE. We begin by defining the \emph{neighborhood}.
\begin{definition}[Neighborhood of a set $X$]
  For each $X \subseteq V$ we define the \emph{neighborhood} of $X$ as
  $$
  N(X) = \{ Y \subseteq V \setminus X \mid X \cup Y \in E \}
  $$

  For each $j \geq 0$ we define the \emph{$j$-neighborhood} of $X$ as
$$
N_j(X) = \{ Y \in N(X) \mid  |Y| = j \} = \{ Y \subseteq V \setminus X \mid X \cup Y \in E \wedge |Y| = j \}
$$

We define the \emph{$j$-degree} of $X$ as $D_j(X) = |N_j(X)|$.
\end{definition}

This definition generalizes the usual definition of neighborhood, in that if $G$ is an ordinary graph and $v$ is a vertex with edges $\{v, w_1 \}, \dots, \{v, w_k \}$, then $N_1( \{ x \}) = \{ \{w_1 \}, \dots, \{w_k \} \}$. Note that if $X$ is an edge of $G$, then $N_0(X) = \{ \emptyset \}$, while if $X$ is not an edge of $G$ then $N_0(X) = \emptyset$. 

We will show that, for an appropriate choice of $p$, the ``normalized degree'' of the hypergraph (defined in an appropriate way), reduces by a constant factor. We define this measure in terms of a single scalar parameter $\Delta$ as follows:\footnote{This definition, in particular the choice of function of $g$, may seem unmotivated here. Please see the proof of Proposition~\ref{Fequil-prop1}, and the remark following that proof, for explanation of why $g$ is defined in this way. We contrast this with the definition used in \cite{kelsen}, whose degree condition can be interpreted as  $D_j \leq (\Delta (\log n)^{-f(j+|X|)})^j$ for a different function $f$.
}
\begin{definition}[$\Delta$-constrained hypergraph]
  \label{delta-def}
Given a hypergraph $G$ and a real number $\Delta \geq 0$, we say that $G$ is \emph{$\Delta$-constrained} if for each $X \subseteq V$ with $0 < |X| < r$ and each integer $j \geq 1$, we have
$$
D_j(X) \leq \Delta^j (\log m)^{-g(j+|X|)}
$$
where throughout this paper we define
$$
g(\ell) = 2^{\ell + 2} - 9
$$
for integers $\ell \geq 2$.
\end{definition}

During the evolution of the algorithm, we will not be able to precisely maintain that the hypergraph is $\Delta$-constrained. Instead, we maintain a slightly weaker property we refer to as \emph{$\Delta$-semiconstrained}:
\begin{definition}[$\Delta$-semiconstrained hypergraph]
  \label{delta-def2}
Given a hypergraph $G$ and a real number $\Delta \geq 0$, we say that $G$ is \emph{$\Delta$-semiconstrained} if for each $X \subseteq V$ with $0 < |X| < r$ and each integer $j \geq 1$, we have $D_j(X) \leq 2 \Delta^j (\log m)^{-g(j+|X|)}$.
\end{definition}

\subsection{Algorithm overview}
In this section, we provide an overview of why the REDUCE algorithm works, given that the input hypergraph is $\Delta$-constrained and $p$ is fixed to $p = 1/\Delta$.  The main idea is to show that the intermediate graphs $G^{(t)}$ are $\Delta$-semiconstrained and the final hypergraph $G^{(T+1)}$ is $\tfrac{1}{2} \Delta$-constrained.

To show these desired properties, we analyze the change in $D^{(t)}_j(X)$ over time for some arbitrary $X \subseteq V$ and $j \geq 1$.  We will describe two main proceses that change $D_j(X)$ significantly: \emph{collapse} and \emph{edge-migration}. (There are other processes that can also change $D_j(X)$, but they do not appear to have a significant effect on the algorithm.)

We say that $X$ \emph{collapses at round $t$} if some $Y \in N^{(t)}(X)$ is added to the independent set, i.e. $Y \subseteq A^{(t)}$. In this case, observe that $X$ (or some $Z \subseteq X$) will appear as an edge in $G^{(t+1)}$. Since the edge set of the residual hypergraph $G^{(t')} = G^{\text{res}}_I$ does not contain nested edges, this further implies that $N^{(t')}_j(X) = \emptyset$ for  $t' > t, j \geq 1$.  Thus, the collapse phenomenon tends to \emph{decrease} $D^{(t)}_j(X)$. It is the main way that the algorithm makes progress toward reducing the degree.

On the other hand, we say that a set $Y \in N_k^{(t)}(X)$ \emph{migrates to $N_j^{(t+1)}(X)$} if there is some subset $Z \subseteq Y$ of size $k - j$ with $Z \subseteq A^{(t)}$.  In this case, the new hypergraph $G^{(t+1)}$ has a new smaller edge $X \cup (Y \setminus Z)$, which contributes to $D^{(t+1)}_j(X)$ instead of $D_k^{(t)}(X)$. This can \emph{increase} $D^{(t+1)}_j(X)$, which is fundamentally different from graph MIS algorithms where the degrees decrease monotonically.

For $1 \leq j < k \leq r$, we let $M^{(t)}_{j,k}(X)$ denote the number of sets $Y \in N^{(t)}_k(X)$ migrating to $N^{(t+1)}_j(X)$. This is the only way in which $D_j(X)$ can increase; thus, it holds that
\begin{equation}
  \label{djt-bnd}
  D_j^{(t+1)}(X) \leq D_j^{(t)}(X) + \sum_{k = j+1}^r M^{(t)}_{j,k}(X)
  \end{equation}

There are other factors which can decrease $D_j(X)$ (for example, migration from $N_j(X)$ to $N_i(X)$ for $i < j$); thus, (\ref{djt-bnd}) is just an upper bound. However, all other changes to $D_j(X)$ appear to be much smaller in magnitude than edge-collapse and migration to $D_j(X)$.

Whenever we discuss migration $M_{j,k}^{(t)}(X)$, we implicitly assume that $1 \leq j < k \leq r$ and $1 \leq |X| \leq r - k$. Likewise, when we discuss the degree $D_j^{(t)}(X)$, we implicitly assume that $j \geq 1$ and $1 \leq |X| \leq r - j$.

These two phenomena have very different behaviors. As we will show, the migration to $N_j(X)$ is concentrated, and with high probability it is always small. Thus, it causes $D_j^{(t)}(X)$ to slowly but steadily increase at essentially a constant rate. On the other hand, collapse is an all-or-nothing effect: it only occurs with a low probability but when it does so it causes $D_j^{(t)}(X) = 0$. Also, the collapse phenomenon becomes more likely when $D_j^{(t)}(X)$ is large.

Let us see how these two effects combine to reduce the degree, after an initial burn-in period. If $D_j^{(t)}(X)$ is large for some round $t \gg 1$, then it must have been large for many prior rounds (since $D_j^{(t)}(X)$ can only grow slowly). This is unlikely, because in each such round, it would have a large probability of collapsing. Thus, contrari-wise, $D_j^{(t)}(X)$ is small with high probability. Note that the burn-in period, during which $D_j^{(t)}(X)$ may grow slightly from its initial value of $\Delta^j (\log m)^{-g(j+x)}$, is why we can only expect that the intermediate graphs $G^{(t)}$ are $\Delta$-semiconstrained instead of $\Delta$-constrained.

Let us state formally the two bounds regarding the effects of collapse and edge-migration:
\begin{proposition}
\label{Fcollapse-prop}
Suppose that $G^{(t)}$ is $\Delta$-semiconstrained. Then for any $X \subseteq V$, we have $\Pr( \text{$X$ collapses at round $t$}) \geq  \tfrac{1}{4} \sum_{k=1}^r \Delta^{-k} D^{(t)}_k(X)$
\end{proposition}

\begin{proposition}
\label{Fmigrate-prop2}
Suppose $G^{(t)}$ is $\Delta$-semiconstrained. For every set $X \subseteq V$ and $j \geq 1$ we have $D_j^{(t+1)}(X) \leq D_j^{(t)}(X) + \Delta^j (\log m)^{2.02 - g(j+1+x)}$ wvhp.
\end{proposition}

The proof of Proposition~\ref{Fcollapse-prop} is nearly identical to the analysis of Beame \& Luby \cite{beame-luby}, with slightly different parameters; we defer it to Appendix~\ref{collapse-proofs}. We show Proposition~\ref{Fmigrate-prop2} in Section~\ref{Fedge-migrate-sec}. In Section~\ref{balance-sec}, we analyze the equilibrium behavior between the two processes, showing that the REDUCE algorithm reduces the degree of $G$ by a constant factor. Finally, we show in Section~\ref{Ftogether-sec} how a simple induction on $\Delta$ produces the full MIS.

\subsection{Assumptions on parameter sizes}
Throughout our analysis of REDUCE (and, later, the derandomized version of REDUCE), we make a number of assumptions on the parameter sizes. We will later discuss how to enforce these bounds  by some simple preprocessing steps.
\begin{enumerate}
  \item[(R1)] $n$ is larger than any needed constants.
  \item[(R2)] $n \leq m \leq n^r$
  \item[(R3)] $2 \leq \Delta \leq m^5$, and $\Delta$ is a power of $2$
  \item[(R4)] $2^{r} < \frac{\log n}{\log \log n}$
  \item[(R5)] $p = 1/\Delta$
\end{enumerate}

We let $s = \log_2 \Delta$; by (R3), $s$ is a non-negative integer in the range $1$ to $5 \log_2 m$. The condition (R4) is not intuitive; the explanation for this is that our goal is to get an algorithm running in $(\log n)^{2^{r+3}}$ time; if condition (R4) is violated, then $(\log n)^{2^{r+3}} \geq n$ and this can be achieved by the trivial sequential algorithm.

We also remark that the deterministic algorithm (which we develop later) requires a small amount of slack in its parameters compared to the randomized algorithm. For consistency, the parameters have been selected for the deterministic algorithm. A more careful analysis could show a slightly lower runtime of $(\log n)^{2^r + O(1)}$ for the randomized algorithm. However, we have made no serious effort to optimize our constant terms and we do not pursue this here.
  
We would also like to briefly discuss our proof strategy: since the formulas we encounter rapidly become very complicated, our overriding goal is to \emph{reduce the number of terms.} For this reason, we often prefer to use small constants such as $0.01$ instead of explicit $\epsilon$ quantities. Also,  we often upper-bound constant terms by expressions $(\log m)^{0.01}$; this allow us to collect a number of disparate contributions into a single term of the form $(\log m)^c$ for constant $c$.

\subsection{Edge migration: proof of Proposition~\ref{Fmigrate-prop2}}
\label{Fedge-migrate-sec}
In this section, we show a concentration bound for the random variable $M^{(t)}_{j,k}$. The main innovation of Kelsen \cite{kelsen} was to derive an exponentially-strong bound for this random variable. Since that paper (and partly as a result of that paper), the topic of polynomial concentration bounds has received much more attention. Our analysis will use a result of Schudy \& Sviridenko \cite{schudy-sviridenko}, which we state in a form specialized to Bernoulli random variables:
\begin{theorem}[\cite{schudy-sviridenko}]
\label{schudy-thm}
Let $S(x_1, \dots, x_n)$ be a degree-$q$ polynomial in $n$ variables, of the form 
$$
S(x_1, \dots, x_n) = \sum_{\substack{Z \subseteq [n] \\ |Z| \leq q}} a_Z \prod_{i \in Z} x_i
$$
for non-negative real coefficients $a_Z$.

Let $X_1, \dots, X_n$ be independent random variables wherein each $X_i$ is distributed as Bernoulli-$p_i$. For $\ell = 0, \dots, q$  define
$$
\mu_{\ell} = \max_{\substack{W \subseteq [n] \\ |W| = \ell}} \sum_{Z \supseteq W} a_Z \prod_{i \in Z \setminus W} p_i
$$

Note that $\mu_{\ell}$ can be interpreted as the maximum expected value of any partial $\ell$-order derivative of the polynomial $S$. 

Then for any $\lambda \geq 0$ we have
$$
\Pr( |S(X) - \bE[S(X)]| \geq \lambda) \leq \exp\Bigl( 2 - \min \Bigl( \min_{\ell=1, \dots, q} \frac{\lambda^2}{\mu_0 \mu_{\ell} R^q}, \min_{\ell=1, \dots, q} \bigl(  \frac{\lambda}{\mu_{\ell} R^q} \bigr)^{1/\ell} \Bigr) \Bigr)
$$
where $R \geq 1$ is some universal constant.
\end{theorem}

\begin{proposition}
\label{Fmigrate-prop1}
Suppose that $G^{(t)}$ is $\Delta$-semiconstrained, and let $X \subseteq V$ with $|X| = x \geq 1$. For any $1 \leq j < k$ we have $M^{(t)}_{j,k}(X) \leq c^{k-j} 2^k \Delta^j (\log m)^{1.01(k-j) - g(k+x)}$ wvhp for some constant $c > 0$.
\end{proposition}
\begin{proof}
We omit the superscript $(t)$ for readability. Suppose that $Y \in N_k(X)$. A necessary condition for $Y$ to migrate to $N_j(X)$ is for $C(Z) = 1$  for some $Z \subseteq Y$ with $|Z| = k - j$. Thus, if we define
$$
S = \sum_{Y \in N_k(X)} \sum_{\substack{Z \subseteq Y \\ |Z| = k - j}} C(Z) = \sum_{\substack{Z \subseteq V \setminus X \\ |Z| = k-j}} C(Z) D_j(X \cup Z)
$$
then $M_{j,k}(X) \leq S$. Here $S$ is a polynomial of degree $q = k - j$ of the form specified by Theorem~\ref{schudy-thm}, and the independent Bernoulli random variables are $C(z)$ for $z \in V \setminus X$ and the coefficients are $a_Z = D_j(X \cup Z)$ for sets of size $|Z| = q$.

We first show an upper bound on $\mu_\ell$ for each $\ell = 0, \dots, q$. To do so, consider any $W \subseteq V \setminus X$ with $|W| = \ell$. As $G$ is $\Delta$-semiconstrained, we have
\begin{align*}
  \sum_{Z \supseteq W} a_Z \prod_{i \in Z \setminus W} p_i &= p^{q - \ell} \sum_{\substack{Z \supseteq W \\ |Z| = q}} D_j(X \cup Z) = (1/\Delta)^{q - \ell} D_{k-\ell} (X \cup W) \tbinom{k - \ell}{j} \\
  &\leq 2 \tbinom{k-\ell}{j} (1/\Delta)^{q-\ell} \Delta^{k-\ell} (\log m )^{-g(k+x) } \leq 2^{k+1} \Delta^j (\log m)^{-g(k+x)} 
\end{align*}

Thus, we have the bound
\begin{equation}
\label{Fmu-bound}
\mu_{\ell} \leq 2^{k+1} \Delta^j (\log m)^{-g(k+x)}
\end{equation}

We wish to apply Theorem~\ref{schudy-thm} with some choice $\lambda > 0$ to obtain a probability $e^{-\Omega(\log^{1.01} m)}$ that $S(X) \geq \bE[S(X)] + \lambda$. It suffices to show that
$$
\min \Bigl( \min_{i=1, \dots, q} \frac{\lambda^2}{\mu_0 \mu_i R^q}, \min_{i=1, \dots, q} \bigl(  \frac{\lambda}{\mu_i R^q} \bigr)^{1/i} \Bigr) \geq \log^{1.01} m
$$
Substituting the bound (\ref{Fmu-bound}), it suffices to show the following bounds for $i = 1, \dots, q$:
\begin{align*}
\lambda &\geq R^{\tfrac{1}{2}(k-j)} 2^{k+1} \Delta^j (\log m)^{0.505 - g(k+x)} \\
\lambda &\geq R^{k-j} 2^{k+1} \Delta^j (\log m)^{1.01 i - g(k+x)} 
\end{align*}
and so it suffices to satisfy the second condition at $i = k-j$. With this choice of $\lambda,$ we have wvhp 
$$
S \leq \mu_0 +  R^{k-j} 2^{k+1}  \Delta^j (\log m)^{1.01(k-j) - g(k+x)} \leq c^{k-j} 2^k \Delta^j (\log m)^{1.01(k-j) - g(k+x)} 
$$
for some constant $c > 0$.
\end{proof}

\begin{repproposition}{Fmigrate-prop2}
Suppose $G^{(t)}$ is $\Delta$-semiconstrained. For every set $X \subseteq V$ and $j \geq 1$ we have $D_j^{(t+1)}(X) \leq D_j^{(t)}(X) + \Delta^j (\log m)^{2.02 - g(j+1+|X|)}$ wvhp.
\end{repproposition}
\begin{proof}
  By applying  Proposition~\ref{Fmigrate-prop1} for each $k = j+1, \dots, r$ and taking a union bound over all such values of $k$, we see that wvhp $M^{(t)}_{j,k} \leq c^{k-j} 2^k \Delta^j (\log m)^{1.01(k-j) - g(k+x)}$   for all such values of $k$. Now sum over $k > j$ to get:
  \begin{align*}
    D_j^{(t+1)}(X) &\leq D_j^{(t)}(X) + \sum_{k =j+1}^r M^{(t)}_{j,k} \leq D_j^{(t)}(X) + \sum_{k=j+1}^{\infty} c^{k-j} 2^k \Delta^j (\log m)^{1.01(k-j) - g(k+x)}
  \end{align*}
  In this sum, the ratio between the $(k+1)^{\text{th}}$ term and $k^{\text{th}}$ term is
  $$
  2 c (\log m)^{1.01 + g(k+x) - g(k+x+1)} = 2 c (\log m)^{1.01 - 2^{k+x+2}} \leq 2 c (\log m)^{1.01 - 2^{2+1+2}} \leq 1/2
    $$

  Thus, the summands decrease exponentially and the sum can be bounded by twice the summand at $k=j+1$. As $j \leq r \leq \log_2 ( \frac{\log n}{\log \log n})$ this implies that wvhp
  $$
      D_j^{(t+1)}(X) \leq D_j^{(t)}(X) + 2^{j+2}  c \Delta^j (\log m)^{1.01 - g(j+1+x)} \leq D_j^{(t)}(X) + \Delta^j (\log m)^{2.02 - g(j+1+x)}
$$

To complete the proof, take a union bound over all sets $X \subseteq V$ of size at most $r$ and all integers $j = 1, \dots, r$. By our choice of $r$ there are at most $n^{\log \log n} \ll e^{-(\log m)^{1.01}}$ choices for such $X, j$.
\end{proof}

\subsection{The balance between edge-migration and collapse}
\label{balance-sec}
We now analyze the interaction between the two phenomena affecting $D_j(X)$, and how in combination they lead to an overall decrease in $D_j(X)$.

\begin{proposition}
  \label{Fbound-prop}
  Let $X \subseteq V$. Then for any integers $j \geq 1, 1 \leq \tau < t \leq T$, the following event occurs wvlp:
\begin{enumerate}
\item $G^{(i)}$ is $\Delta$-semiconstrained for $i =  t - \tau, \dots, t$
  \item $D_j^{(t+1)}(X) > \frac{\Delta^j (\log m)^{1.01}}{\tau} + \tau \Delta^j (\log m)^{2.02 - g(j+1+|X|)}$.
\end{enumerate}
\end{proposition}
\begin{proof}
  Let $|X| = x $ and let $\gamma =  \frac{\Delta^j (\log m)^{1.01}}{\tau} + \Delta^j \tau (\log m)^{2.02 - g(j+1+x)}$. By Proposition~\ref{Fmigrate-prop1}, it holds wvhp that at any time $t$ where $G^{(t)}$ is $\Delta$-semiconstrained, we have $D_j^{(t+1)} (X) \leq D_j^{(t)}(X) + \Delta^j (\log m)^{2.02 - g(j+1+x)}$. Thus, barring wvlp events, if $D_j^{(t+1)}(X) > \gamma$ then every round $i = t - \tau + 1, \dots, t$ satisfies
  $$
  D_j^{(i)} (X) > \gamma - (t+1-i) \Delta^j (\log m)^{2.02 - g(j+1+x)} \geq \gamma - \tau \Delta^j (\log m)^{2.02 - g(j+1+x)} = \frac{\Delta^j (\log m)^{1.01}}{\tau}
  $$

  In each such round $i$, Proposition~\ref{Fcollapse-prop} shows that $X$ would collapse with probability at least $\tfrac{1}{4} D_j^{(i)} (X) \Delta^{-j}$. Our bound on $D_j^{(i)}$ shows that this is at least $\frac{(\log m)^{1.01}}{4 \tau}$. In order for the stated event to occur, $X$ must avoid collapse at all such times, which would have probability
\[
\Pr( \text{$X$ avoids collapse by round $t$} ) \leq (1-\frac{(\log m)^{1.01}}{4 \tau})^{\tau} \leq \exp(-\tfrac{1}{4} (\log m)^{1.01}) \qedhere
\]
\end{proof}

\begin{proposition}
\label{Fequil-prop1}
Wvhp $G^{(1)}, \dots, G^{(T)}$ are $\Delta$-semiconstrained.
\end{proposition}
\begin{proof}
For any $X \subseteq V$ of size $|X| = x$ with $1 \leq x \leq r$, any integer $j = 1, \dots, r$, and any integer $t = 1, \dots, T$ define the event $\mathcal B(X,j,t)$ as follows:
\begin{enumerate}
\item $G^{(1)}, \dots, G^{(t)}$ are $\Delta$-semiconstrained
\item $D^{(t+1)}_j(X) > 2 \Delta^j (\log m)^{-g(j+|X|)}$
\end{enumerate}

If none of the events $\mathcal B(X,j,t)$ occur, then $G^{(1)}, \dots, G^{(T)}$ are all $\Delta$-semiconstrained. There at most $n^r$ choices for $X$ and $T \leq n$ choices for $t$, and $r \leq O(\log \log m)$ choices for $j$; overall, the total number of choices is at most $e^{O(\log m \log \log m)}$. By a union bound, it thus suffices to show that $\mathcal B(X,j,t)$ occurs wvlp for a fixed $X, j, t$.

So let us fix $X, j, t$ with $|X| = x \geq 1$ and define
$$
\tau = (\log m)^{g(j+x+1) - g(j+x) - 2.02}
$$

If $t \leq \tau$, then note that the input hypergraph $G = G^{(1)}$ is $\Delta$-constrained and by Proposition~\ref{Fmigrate-prop1} $D_j(X)$ increases by at most $(\log m)^{2.02 - g(j+1+x)}$ per round wvhp. So 
\begin{align*}
D_j^{(t)}(X) &\leq \Delta^j (\log m)^{-g(j+x)} + t \Delta^j (\log m)^{2.02 - g(j+1+x)} \leq 2 \Delta^j (\log m)^{-g(j+x)}
\end{align*}
and so the event $\mathcal B(X,j,t)$ does not occur.

Next, suppose that $t > \tau$. Then by Proposition~\ref{Fbound-prop} we have wvhp
\begin{align*}
  D_j^{(t)}(X) &\leq \frac{\Delta^j (\log m)^{1.01}}{\tau} + \tau \Delta^j (\log m)^{2.02 - g(j+1+x)}  \\
  &= \Delta^j (\log m)^{3.03 - g(j+1+x) + g(j+x)} + \Delta^j (\log m)^{-g(j+x)} \\
  &= \Delta^j (\log m)^{-g(j+x)} \bigl( (\log m)^{3.03 - g(j+1+x) + 2 g(j+x)} + 1 \bigr)
\end{align*}

Substituting the value of $g$ we see that $g(j+1+x) + 2 g(j+x) = 2^{j+x+3} - 9 + 2 (2^{j+x+2} - 9) = - 9$, and so
$$
D_j^{(t)}(X) \leq \Delta^j  (\log m)^{-g(j+x)} \bigl( (\log m)^{-5.97} + 1 \bigr) \leq 2 \Delta^j (\log m)^{-g(j+x)}
$$
and so $\mathcal B(X,j,t)$ again does not occur.
\end{proof}

\noindent \textbf{Remark on our choice for the function $g(\ell)$.} In light of Proposition~\ref{Fequil-prop1}, we can discuss our choice for the function $g$ used in Definition~\ref{delta-def}. In order to get the required bound on $D_j^{(t)}$ for the case $ t > \tau$, we need the bound
$$
3.03 - g(j+1+x) + 2 g(j+x) \leq 0
$$
whence we derive that $g(\ell) \geq 2 g(\ell-1) + 3.03$. Solving this recurrence relation, we see that $g(\ell)$ should have the form $g(\ell) = a 2^{\ell} + b$ for constants $a, b$. The precise constant terms used in our definition (namely $g(\ell) = 2^{\ell+2} - 9$) are chosen because the deterministic algorithm will need a slight additional slack in its constraints.

\begin{theorem}
\label{Fdecrease-thm}
If $G$ is $\Delta$-constrained then wvhp, $G^{(T+1)}$ is $\tfrac{1}{2} \Delta$-constrained.
\end{theorem}
\begin{proof}
Consider some set $X \subseteq V$ and integer $j \geq 1$ with $|X| = x$ for $1 \leq x \leq r-j$. Let us set 
  $$
  \tau = \tfrac{1}{2} 2^{-j} (\log m)^{g(j+x+1) - g(j+x) - 2.02}
  $$

Observe that $\tau \leq (\log m)^{2^{j+x+2}} \leq (\log m)^{2^{r+2}} = T$.  By Proposition~\ref{Fequil-prop1}, wvhp the graphs $G^{(1)}, \dots, G^{(T)}$ are all $\Delta$-semiconstrained.  So Proposition~\ref{Fmigrate-prop2} applied at $t = T$ shows that wvhp
  \begin{align*}
  D_j^{(T+1)}(X) &\leq \frac{\Delta^j (\log m)^{1.01}}{\tau} + \tau \Delta^j (\log m)^{2.02 - g(j+1+x)} \\
  &= 2^{j+1} \Delta^j (\log m)^{3.03 + g(j+x) - g(j+x+1)} + 2^{-j-1} \Delta^j (\log m)^{-g(j+x)} \\
  &= (\tfrac{\Delta}{2})^j (\log m)^{-g(j+x)} \Bigl( 2^{2j +2} (\log m)^{3.03 + 2 g(j+x) - g(j+x-1)} + \tfrac{1}{2} \Bigr)
  \end{align*}

As $j \leq r \leq \log_2 \bigl( \frac{\log n}{\log \log n} \bigr)$, we have $2^{2j + 2} \leq (\log m)^2$. Substituting the value of $g$ shows that $3.03 + 2 g(j+x) - g(j+x-1) = -5.97$. Therefore, wvhp we have
$$
D_j^{(T+1)}(X) \leq (\tfrac{\Delta}{2})^j (\log m)^{-g(j+x)} \bigl( (\log m)^{2 - 5.97} + \tfrac{1}{2} \bigr) \leq (\tfrac{\Delta}{2})^j (\log m)^{-g(j+x)}
$$

Finally take a union bound over sets $X$ of size at most $r$ and integers $j$.
\end{proof}

\subsection{Producing the full MIS}
\label{Ftogether-sec}
So far, we have studied a single application of REDUCE, where $p$ is fixed to $1/\Delta$. The full MIS algorithm allows $p$ to change over time.

{
\renewcommand{\thetheorem}{\ref{Fdet-thm-overall}}
\begin{theorem}  
There is a randomized algorithm using $O(\log m) + (\log n)^{2^{r+3} + O(1)}$ expected time and $O(n + m \log n)$ processors, to produce an MIS of a rank-$r$ hypergraph.
\end{theorem}
}
\begin{proof}
  If $2^{r} > \frac{\log n}{\log \log n}$ or $n$ is smaller than any needed constants, then use the sequential algorithm; this will run in time $O(n)$, which is smaller than $O((\log n)^{2^r})$.  If $m < n$, then we run the following simple pre-processing step in $O(\log(mn))$ time: for each edge $e \in G$, mark one vertex arbitrarily from $e$; then add all unmarked vertices to the independent set. Hence we assume $m \geq n \gg 1$.

After these checks are completed, we use the following algorithm FIND-MIS, which takes as input a rank-$r$ hypergraph $G$ and returns an MIS of $G$ wvhp. To bound the expected runtime, we can simply run the algorithm multiple times until it succeeds.
  
\begin{algorithm}[H]
\centering
\begin{algorithmic}[1]
\State{Let $\Delta_0 = 2^{\lfloor \log_2 m^5 \rfloor}$}
\State Let $I = \emptyset$. 
\State \textbf{for }{$i = 0, 1, 2, \dots, \log_2 \Delta_0$} \textbf{do} update $I \leftarrow I \cup \text{REDUCE}(G^{\text{res}}_I, \frac{2^i}{\Delta_0})$
\State Output $I$
\end{algorithmic}
\caption{The algorithm FIND-MIS($G$) 
\label{algo:find-mis}}
\end{algorithm}

Let $G_i$ be the residual hypergraph before the $i^{\text{th}}$ iteration of the loop (so $G_0$ is the original input hypergraph). We claim that for $i = 0, \dots, \log_2 \Delta_0$, the hypergraph $G_i$ is $\Delta_i = \Delta_0 2^{-i}$ constrained. For the base case of the induction, observe that $D_j(X) \leq m$ trivially; so it suffices to show that for any integers $j, x \geq 1$ with $j + x \leq r$ we have $\Delta_0^j (\log m)^{-g(j+x)} \geq m$. To show this, note that for $j + x \leq r$ we have:
\begin{align*} 
  \Delta_0^j (\log m)^{-g(j+x)} &\geq \frac{m^5}{2}  (\log m)^{-2^{r+2} + 9} \geq \frac{m^5}{2} (\log m)^{-\frac{4 \log n}{\log \log n} + 9} \geq m^5 (\log m)^{-\frac{4 \log m}{\log \log m}} = m
\end{align*}

For the induction step, we apply Theorem~\ref{Fdecrease-thm}. The conditions (R1), (R2), (R4) are satisfied by our preprocessing steps and (R3) is satisfied by our choice of $\Delta_i = \frac{\Delta_0}{2^i}$.

There are $1 + \log_2 \Delta_0 = O(\log m)$ applications of REDUCE, each of which takes $(\log m)^{2^{r+2} + O(1)}$ time. So the overall runtime is at most $(\log m)^{2^{r+2} + O(1)}$. Since $m \leq n^r$, this in turn is at most $(\log n)^{2^{r+3}+O(1)}$.

Now consider $G' = G_{\log_2 \Delta_0}$. We have shown above that $G'$ is $1$-constrained. We claim that $G'$ has rank at most $1$. For, suppose that $e$ is some edge of rank $k \geq 2$, and let $x \in e$ be an arbitrary vertex. Set $X = \{x \}$ and $j = k - 1$. In this case, we have $\{ e \setminus X \} \in N_{j} (X )$ and so $D_{j}( X) > 0$. But since $G'$ is $1$-constrained it holds that $D_j(X) \leq (\log m)^{-g(j+|X|)} \leq (\log m)^{-g(2)} = (\log m)^{-7} < 1$.

Since $G' = G^{\text{res}}_I$ has rank 1, it is trivial to extend $I$ to an MIS in $O(\log n)$ time.
\end{proof}

\textbf{A distributed algorithm.} Kutten et. al. \cite{kutten} described a number of distributed algorithms for hypergraph MIS. These are also based on the REDUCE procedure, but there is a complication: in Kelsen's original algorithm, the probability $p$ used in REDUCE is a function of the maximum degree of the current residual hypergraph. This is a global statistic, and so it cannot be computed easily in a distributed (local) algorithm. As a result, the algorithm of \cite{kutten} is fairly complex and has some additional technical limitations.

By contrast, the probability $p$ used in Algorithm~\ref{algo:find-mis} depends solely on $n,r$ and the number of iterations since the algorithm has begun. These are globally known parameters and so each step of REDUCE and of Algorithm~\ref{algo:find-mis} can easily be simulated in $O(1)$ distributed rounds. Thus, Algorithm~\ref{algo:find-mis} can be trivially converted into a distributed algorithm, giving the following result:
\begin{theorem}
There is a $(\log n)^{2^{r+3} + O(1)}$-round randomized algorithm to compute an MIS of a rank-$r$ hypergraph wvhp in the CONGEST distributed computing model.
\end{theorem}

\section{Derandomizing the algorithm}
We now turn to derandomizing the REDUCE algorithm in the PRAM setting. Instead of drawing $C^{(t)}(1), \dots, C^{(t)}(n)$ randomly at each round $t$, we do so by a careful deterministic process. This requires two major modifications of the algorithm.

First, in order to achieve a greater degree of independence, we do not select the vector $C$ all at once. Each entry $C^{(t)}(v)$ is a Bernoulli random variable with mean $\tfrac{1}{\Delta}$, which can also be viewed as the conjunction of $s = \log_2 \Delta$ independent fair coins $B^{(t,1)}(v), \dots, B^{(t,s)}(v)$. Accordingly, we build $C$ through $s$ sub-rounds. In sub-round $i = 1, \dots, s$, we select the bit-vector $(B^{(t,i)}(v) \mid v \in V)$ from the support of a distribution $\Omega$ (to be specified), and at the end  we set
$$
C^{(t)}(v) = B^{(t,1)}(v) \wedge B^{(t,2)}(v) \wedge \dots \wedge B^{(t,s)}(v)
$$

For each round $t$ and sub-round $i$ we define $B^{ \leq (t,i)}$ to be the vector $$
B^{\leq(t,i)} = ( B^{(1,1)}, \dots, B^{(1,s)}, B^{(2,1)}, \dots, B^{(2,s)}, \dots, B^{(t,1)}, \dots, B^{(t,i)} ).
$$
Note that $B^{(\leq t,0)} = B^{\leq(t-1,s)}$.

The second modification is to select $B^{(t,i)}$ based on the method of conditional expectations. We define a series of potential functions $\Phi_{\sigma}^{(t,i)}$, which are functions mapping  $B^{\leq (t,i)}$ to non-negative real numbers; each of these represents (an approximation to) the probability of a certain bad-event, ``as if'' the bits $B^{\leq(t,i)}$ are fixed to some arbitrary value and the remaining bits $B^{(t,i+1)}, \dots, B^{(T,s)}$ are drawn independently. The index $\sigma$ is a label for some bad-event that we want to avoid, for example, $\sigma$ may be a label that some set $X$ fails to collapse at some given time $t$.  For any fixed index $\sigma$, we let $\Phi_{\sigma}$ denote the collection of functions $\Phi_{\sigma}^{(t,i)}$, and we refer to this family as a \emph{potential summand}.  

We also define an overall potential function $\Phi^{(t,i)} = \sum_{\sigma} \Phi_{\sigma}^{(t,i)}$. At each stage, we select a value for $B^{(t,i)}$ in order to minimize $\Phi^{(t,i)}$. See Algorithm~\ref{alg-det}.

\begin{algorithm}[H]
\centering
\caption{DET-REDUCE($G$) 
\label{alg-det}}
\begin{algorithmic}[1]
  \State Generate the probability space $\Omega$ over ground set $V$ satisfying condition (Q) (see below)
\For{$t = 1, \dots, T = (\log m)^{2^{r+2}}$}
\State Let $G^{(t)} = G^{\text{res}}_I$.
\For{$i = 1, \dots, s$}
\State{In parallel, search over all possible $b \in \text{support}(\Omega)$.}
\State{Choose the value $B^{(t,i)} = b$ to minimize $\Phi^{(t,i)}( B^{\leq (t,i)} )$}
\EndFor
\State{For each $v \in V$, set $C^{(t)}(v) = B^{(t,1)}(v) \wedge \cdots \wedge B^{(t,s)}(v)$}
\State Form the set $A^{(t)} = \{ v \in V \mid C^{(t)} (v) = 1 \wedge C^{(t)}(Y) = 0 \text{ for all $Y \in N^{(t)} ( \{ v \})$} \}$.
\State{Update $I \leftarrow I \cup A^{(t)}$.}
\EndFor
\State{Set $G^{(T+1)} = G^{\text{res}}_I$ and output $I$}
\end{algorithmic}
\end{algorithm}

Specifically,  the probability space $\Omega$ is chosen to satisfy the following condition:
  \begin{enumerate}
\item[(Q)] For any integer $k$ with $0 \leq k \leq L \log m$ (where $L$ is a large constant to be specified later), and any indices $1 \leq i_1 < i_2 < \dots < i_k \leq n$, and any $y \in \{0, 1 \}^k$,  we have
$$
(1 - \frac{1}{2s}) 2^{-k} \leq \Pr_{X \sim \Omega} (X_{i_1} = y_1 \wedge X_{i_2} = y_2 \wedge \dots \wedge X_{i_k} = y_k) \leq (1+\frac{1}{2s}) 2^{-k} 
  $$
  \end{enumerate}

  This condition is known as a \emph{$\delta$-approximate $w$-wise independence}, where $\delta = \frac{2^{-L \log m}}{2 s}$ and $w = L \log m$. Naor \cite{naor2} showed that such probability spaces exist with support size of order  $2^{O(w + \log(1/\delta) + \log \log n)}$, and be constructed efficiently by deterministic parallel algorithms. With these values of $w$ and $\delta$, and recalling that $s \leq O(\log m)$, we see that $\Omega$ has support size $\poly(m)$.
  
\subsection{Properties of the potential function}

Our proof strategy for DET-REDUCE is to show by induction on $t,i$ that the following condition (I) holds at every round $t = 0, \dots, T$ and sub-round $i = 0, \dots, s$:
\begin{enumerate}
\item[(I)] The  hypergraphs $G^{(t)}$ is $\Delta$-semiconstrained and $\Phi^{(t,i)} < 1$
\end{enumerate}
We also say that \emph{(I) holds up to $(t,i)$} if it holds at every round $t'$ and sub-round $i'$ with either $t' < t$ or $t' = t, i' \leq i$. 
 
In order to carry out this induction and to get an efficient algorithm, we require a number of properties of our family of potential summands. First, we require:
\begin{enumerate}
\item[(P1)] The total number of potential summands $\sigma$ is less than $m^{100}$.
\end{enumerate}

In addition, we require that every index $\sigma$ satisfies the following properties (P2) --- (P5):
\begin{enumerate}
\item[(P2)]  $\Phi^{(\ell,i)}_{\sigma}$ can be computed in parallel using $\poly(m)$ processors and $\polylog(m)$ time.
\item[(P3)] $\Phi_{\sigma}^{(0,s)} \leq m^{-100}$.
\item[(P4)] If (I) holds up to $(\ell, i-1)$ for any $\ell \geq 1$ and $i \geq 1$ and we fix a value for $B^{(\ell, i-1)}$, then the conditional expectation of $\Phi_{\sigma}^{(\ell, i)}$ satisfies
$$
\bE_{B^{(\ell,i)} \sim \Omega} \bigl[ \Phi_{\sigma}^{(\ell,i)}( B^{\leq (\ell,i)}) \mid B^{(\ell,i-1)} \bigr] \leq \Phi_{\sigma}^{(\ell,i-1)}(B^{\leq (\ell,i-1)}).
$$
\item[(P5)] If $G^{(\ell)}$ is $\Delta$-semiconstrained and (I) holds up to round $(\ell-1, s)$ for any $\ell \geq 1$, then with probability one
$$
\Phi_{\sigma}^{(\ell,0)} (B^{\leq (\ell,0)}) \leq \Phi_{\sigma}^{(\ell-1,s)}(B^{\leq (\ell-1,s)}).
$$
\end{enumerate}

Conditions (P1) and (P2) will ensure that DET-REDUCE can be implemented in $\poly(m)$ processors and $\polylog(m)$ time.

Each summand is meant to represent the conditional probability of a certain bad-event. With this interpretation, properties (P3) and (P1) mean that the expected number of occuring bad events is below $1$. Property (P4) is essentially the law of iterated expectations.  (Property (P5) is used for technical reasons, to allowing rescaling between successive times.) Note that, since all our summands are non-negative, condition (I) ensures that $\Phi_{\sigma}^{(\ell, i)} < 1$ for all summands $\sigma$.

There are two types of summands, which we denote (S1) and (S2), to enforce edge-migration and edge-collapse properties similar to the randomized algorithm. We define and analyze (S1) in  Section~\ref{Fedge-migrate-sec-det}. Summand (S2), which is a fairly routine derandomization of the the corresponding randomized algorithm, is mostly covered in Appendix~\ref{collapse-proofs}, with a brief summary in Section~\ref{Fcollapse-det-sec}. We also describe in Section~\ref{Fcollapse-det-sec} how the overall induction proof works to show property (I).

\section{Polynomial concentration bounds via non-central moments}
\label{conc-polynomial}
Just as in the randomized algorithm, we must bound the migration $M^{(t)}_{j,k}(X)$. We do so by derandomizing concentration inequalities for polynomials. The bounds of Section~\ref{Fedge-migrate-sec}, which are derived from \cite{schudy-sviridenko},  are not suitable for this purpose because they are based on  central moments of the underlying  random variables. These may be highly distorted in an approximately-independent probability space.  Our first task is thus to derive new concentration bounds based on non-central moments, which are useful for probability spaces with approximate independence.

We consider a polynomial of the form
$$
S(x_1, \dots, x_n) = \sum_{Z \in U_q} a_Z \prod_{j \in Z} x_j
$$
where $a_Z$ are non-negative real numbers and where we define $U_q$ to be the set of $q$-element subsets of $[n]$, i.e. the set $\binom{[n]}{q}$. For a parameter $\lambda \in [0,1]$, we wish to estimate $\bE[ S(X_1, \dots, X_n)^w ]$, where the variables $X_1, \dots, X_n$ are (approximately) distributed as iid Bernoulli-$\lambda$. As in Theorem~\ref{schudy-thm}, we do so by bounding the expected partial derivatives of $S$.

For each integer $\ell = 0, \dots, q$ we define
\begin{equation}
\label{mu-defn}
\mu_{\ell} =\lambda^{q-\ell} \max_{Y \in U_{\ell}} \sum_{\substack{Z \in U_q \\ Z \supseteq Y}} a_Z 
\end{equation}

Here $\mu_{\ell}$ can be interpreted as the maximum expected value of the partial $\ell$-order derivative of polynomial $S$. Note that this setting is more restricted than that in Theorem~\ref{schudy-thm}: all the sets $Z$ have the same size $q$, and all the variables $X_i$ have the same probability $p_i$, which is equal here to $\lambda$. Thus, this definition (\ref{mu-defn}) is simplified compared to the one used in Theorem~\ref{schudy-thm}.

We get the following concentration inequality for approximate-independence probability spaces:
\begin{theorem}
\label{Fapp-prop2}
Let $a_Z \geq 0$, $\mu_{\ell}$, and polynomial $S$ be defined as above. Suppose $X_1, \dots, X_n$ are Bernoulli random variables which satisfy an approximate $w q$-wise independence condition; namely that for any $k \leq w q$ and any indices $i_1, \dots, i_{k}$ we have
$$
\Pr(X_{i_1} = X_{i_2} = \dots = X_{i_k} = 1) \leq (1+\epsilon) \lambda^k
$$

Then
$$
\bE[S(X_1, \dots, X_n)^w] \leq (1+\epsilon) \Bigl( \sum_{\ell=0}^q \binom{w q}{\ell} \mu_\ell \Bigr)^w
$$
\end{theorem}
\begin{proof}
First expand the sum as:
\begin{align*}
\bE[S(X_1, \dots, X_n)^w] & = \sum_{Z_1, \dots, Z_w \in U_q} a_{Z_1} \dots a_{Z_w} \bE[ \prod_{i \in Z_1 \cup \dots \cup Z_w} X_i ]
\end{align*}

For any $Z_1, \dots, Z_w \in U_q$ the approximate independence condition gives 
$$
\bE[ \prod_{i \in Z_1 \cup \dots \cup Z_w} X_i ] = \Pr( \bigwedge_{i \in Z_1 \cup \dots \cup Z_w} X_i = 1 ) \leq (1+\epsilon) \lambda^{|Z_1 \cup \dots \cup Z_w|}
$$
Hence
$$
\bE[S(X_1, \dots, X_n)^w] \leq (1+\epsilon)  \sum_{Z_1, \dots, Z_w \in U_q} a_{Z_1} \dots a_{Z_w} \lambda^{|Z_1 \cup \dots \cup Z_w|}
$$

To complete the proof, we claim that for $w \geq 0$ we have the inequality
\begin{equation}
\label{Fyy1e}
\sum_{Z_1, \dots, Z_w \in U_q} a_{Z_1} \dots a_{Z_w} \lambda^{|Z_{1} \cup Z_{2} \cup \dots \cup Z_{w}|} \leq \Bigl( \sum_{\ell=0}^q \binom{w q}{\ell} \mu_\ell \Bigr)^w
\end{equation}

We show (\ref{Fyy1e}) by induction on $w$. When $w = 0$, then (\ref{Fyy1e}) is vacuously true as both Left-hand Side (LHS) and Right-hand Side (RHS) are equal to $1$. For $w > 0$, consider the LHS sum in (\ref{Fyy1e}). For a fixed $Z_1, \dots, Z_{w-1}$ we let $A = Z_{1} \cup \dots \cup Z_{w-1}$ and write the sum over $Z_w$ as

{\allowdisplaybreaks
\begin{align*}
\sum_{Z_w \in U_q} a_{Z_w} \lambda^{|Z_{1} \cup Z_{2} \cup \dots \cup Z_{w}|} &= \lambda^{|A|} \sum_{Z_w \in U_q} a_{Z_w} \lambda^{q - |Z_w \cap A|} = \lambda^{|A|} \sum_{\ell=0}^q \sum_{\substack{Z \in U_q \\ |Z \cap A| = \ell}} a_Z \lambda^{q - \ell} \\
 &\leq \lambda^{|A|} \sum_{\ell=0}^q  \sum_{\substack{Y \subseteq A \\ |Y| = \ell}} \sum_{\substack{Z \in U_q \\ Z \cap A = Y}} a_Z \lambda^{q- \ell} \leq \lambda^{|A|} \sum_{\ell=0}^q \sum_{\substack{Y \subseteq A \\ |Y| = \ell}} \sum_{\substack{Z \in U_q \\ Z \supseteq Y}} a_Z \lambda^{q- \ell}\\
 &\leq \lambda^{|A|} \sum_{\ell=0}^q \binom{|A|}{\ell} \mu_\ell \leq \lambda^{|A| } \sum_{\ell=0}^q \binom{w q}{\ell} \mu_\ell.
 \end{align*}
}
Now use the induction hypothesis to compute:
\begin{align*}
&\sum_{Z_1, \dots, Z_w} a_{Z_1} \dots a_{Z_w} \lambda^{|Z_{1} \cup \dots \cup Z_{w}|} \leq \sum_{Z_1, \dots, Z_{w-1}}  a_{Z_1} \dots a_{Z_{w-1}} \lambda^{| Z_{1}  \cup \dots \cup Z_{i_{w-1}}|} \times \sum_{\ell=0}^q \tbinom{w q}{\ell} \mu_\ell \\
&\qquad \qquad \leq \bigl( \sum_{\ell=0}^q \binom{(w-1) q}{\ell} \mu_\ell  \bigr)^{w-1} \times \Bigl( \sum_{\ell=0}^q \binom{w q}{\ell} \mu_\ell  \Bigr) \leq \Bigl( \sum_{\ell=0}^q \binom{w q}{\ell} \mu_\ell  \Bigr)^{w}  \qedhere
\end{align*}
\end{proof}

\section{Edge migration for the deterministic algorithm}
\label{Fedge-migrate-sec-det}
We next convert the probabilistic bounds of Section~\ref{conc-polynomial} into an appropriate pessimistic estimator for the concentration of the edge migration $M^{(t)}_{j,k}$. This is the most technically challenging part of our algorithm. Before we describe the formal construction, let us first provide a high-level overview.

\subsection{Overview}
\label{migrate-overview-sec}
The approach is similar to, and inspired by, a method of Alon \& Srinivasan \cite{alon-srin} to derandomize concentration bounds for sums of independent random variables. This uses two key ingredients: a probability space with approximate independence for $\Theta(\log n)$ variables, and a potential function based on the conditional expectation of high-order moments of the sum. Following their approach, it would be natural to define a polynomial $S$ upper-bounding $M^{(t)}_{j,k}$ and use the conditional expectation of $S(C^{(t)})^w$ as the potential function, for a parameter $w = \Theta(\log n)$.

There is one minor technical detail to keep in mind here, which also appears in the method of Alon \& Srinivasan: since the probability space $\Omega$ only has an approximate independence condition, we cannot guarantee that this potential function decreases at each stage. It may slowly increase, by a factor of $(1+\epsilon)$ at each stage. This can be easily handled by using instead the potential function $(1+\epsilon)^{s-i} \bE[ S(C^{(t)})^w ]$ where the expectation is taken assuming $B^{(t,1)}, \dots, B^{(t,i)}$ are fixed and $B^{(t,i+1)}, \dots, B^{(t,s)}$ are independent Bernoulli-$\tfrac{1}{2}$.

However, there are two severe technical roadblocks to adapting the method of Alon \& Srinivasan to higher-degree polynomials, which are fundamentally different than anything that occurs for sums of independent random variables.

The first roadblock is we must be able to exactly compute the potential function --- an upper bound on it is not enough. Thus, for the method of conditional expectations,  we would need to compute the expectation of $S(C^{(t)})^w$ for some fixed value of $B^{\leq(t,i)}$ . When we fix these bits, then some values of $C^{(t)}$ are forced to be zero while others remain iid Bernoulli with mean $\lambda = 2^{-(s-i)}$. The expectation $S(C^{(t)})^w$ can then be written as a sum over $w$-tuples of the monomials of $S$, i.e. 
\begin{equation}
\label{t1-eqn}
\bE[S(C^{(t)})^w] = \sum_{Z_1 \dots, Z_w} a_{Z_1} \dots a_{Z_w} \lambda^{|Z_1 \cup \dots \cup Z_w|}
\end{equation}

There is no obvious shortcut to computing (\ref{t1-eqn}) other than enumerating over $Z_1, \dots, Z_w$; this would require roughly $m^w$ processors, which is super-polynomial since we are taking $w = \Theta( \frac{\log m}{\log \log m})$. (This is the main difference between our algorithm and the derandomized algorithm of Kelsen \cite{kelsen}, which used $w = \Theta(1)$). To get an efficiently-computable potential function, we use Theorem~\ref{Fapp-prop2} to get the bound
$$
\bE[S(C^{(t)})^w] \leq \Bigl( \sum_{\ell=0}^q \binom{w q}{\ell} \mu_\ell \Bigr)^w
$$
where the terms $\mu_{\ell}$ are  upper bounds on the expeted partial derivatives of $S$. 

Note that $\mu_0$ is the expected zero-order derivative, i.e. $\mu_0 = \bE[S(C^{(t)})]$. In the cases of interest to us, the relevant monomials $Z_i$ are ``mostly'' disjoint. Because of this fact, $\mu_1, \dots, \mu_q$ will be negligible, and hence we expect
$$
\bE[S(C^{(t)})^w] \approx \Bigl( \binom{w q}{0} \mu_0 + \text{low order terms} \Bigr)^w = (\bE[S(C^{(t)})] + \delta)^w
$$

Since these low-order terms (the contributions from the terms $\mu_1, \dots, \mu_q$) are relatively small, we do not need to calculate them exactly; we use a crude upper-bound $\hat \delta$ which, crucially, \emph{does not depend on the values of $C$}. Specifically, we will take our potential function to be $(\bE[S(C^{(t)})] + \hat \delta)^w$, which will be an adequate first-order approximation to $\bE[S(C^t)^w]$.

The second, and related, technical roadblock is much more difficult. This approximation depends on the relative sizes of the (expected) partial derivatives of the polynomial $S$.  The DET-REDUCE process gradually fixes the variables $B^{(t,i)}$; this means that the partial derivatives of the polynomial $S(C^{(t)})$ are also changing. We must ensure that the conditional expectation of the partial derivatives of $S$ changes at approximately the correct rate to keep pace with the changing expectation of $S$ itself. Thus, in parallel to showing concentration bounds for the polynomial $S$, we are forced to show concentration bounds for all of its partial derivatives. These are all mutually interdependent (and at each stage there is some deviation from the mean, leading to some error terms) leading to a complicated recursive formula.

We emphasize that these two hurdles are fundamentally new phenomena which are not present for linear polynomials; deriving appropriate potential functions to handle them requires numerous technical modifications compared to the relatively clean approach of Alon \& Srinivasan.

\subsection{The potential function for edge-migration}
Suppose that induction condition (I) holds up to some fixed round $t$. We will introduce a series of potential summands to collectively guarantee that every set $X$ has $M^{(t)}_{j,k}(X) \approx \Delta^j (\log m)^{-g(k+x)}$, thus ensuring that (I) holds again at round $t+1$. Specifically, our goal will be to show that, for every non-empty set $X$ we have $M_{j,k}^{(t)}(X) \leq \Gamma_{j,k,|X|}$, where we define the parameter
$$
\Gamma_{j,k,x} = \Delta^{j} (\log m)^{-g(k+x) + 4 (k-j)}
$$
for $1 \leq j < k \leq r$ and $1 \leq x \leq r$. To do so, let us define $C^{(t,i)}(v) = \bigwedge_{j=1}^i B^{(t,j)}(v)
$ for a vertex $v$; note that $C^{(t)}(v) = C^{(t,s)}(v)$ and $C^{(t,0)}(v) = 1$. Using this notation, we define the ``partial migration upper-bound'' functions
$$
S^{(t,i)}_{j,k}(X) = \sum_{\substack{Z \subseteq V \setminus X \\ |Z| = k-j}} C^{(t,i)}(Z) D_j^{(t)}(X \cup Z)
$$

We will show by induction on $i$ that
$$
S_{j,k}^{(t,i)}(X) \leq \Gamma^{(i)}_{j,k,|X|} := 2^{(s-i)(k-j)} \Gamma_{j,k,|X|} = 2^{(s-i)(k-j)} \Delta^{j} (\log m)^{-g(k+|X|) + 4 (k-j)}
$$

We emphasize one important feature of this proof: unlike in our analysis of the randomized algorithm, we cannot simply fix a set $X$ and show that $S^{(t,i)}_{j,k}(X)$ is bounded on its own. The quantities $S^{(t,i)}_{j,k}(X)$ are all mutually interdependent, and it is very important to control them simultaneously. To that end, we introduce a collection of potential summands, which we call (S1),  to enforce the bounds $S_{j,k}^{(t,i)}(X) \leq \Gamma^{(i)}_{j,k,|X|}$.

For each non-empty set $X \subseteq V$, each pair of integers $0 < j < k \leq r$, and each round $t = 1, \dots, T$ we have a potential summand $\Phi_{\text{S1}, t, k, j, X}$ as follows:
$$
\Phi^{(\ell, i)}_{\text{S1}, t,k,j,X} = \begin{cases}
  4^{s-i} \Bigl( \frac{S^{(t,i)}_{j,k}}{\Gamma^{(i)}_{j,k,|X|}} +  \frac{1}{\log m}  \Bigr)^w  & \text{if $\ell = t$} \\
m^{-100} & \text{if $\ell < t$} \\
0 & \text{if $\ell > t$}
\end{cases} \qquad \qquad \text{where  $w = \Big \lceil \frac{1000 \log m}{\log \log m} \Big \rceil$ }
$$

Let us briefly describe how this fits into the high-level overview discussed in Section~\ref{migrate-overview-sec}. Up to rescaling, the quantity $S^{(t,i)}_{j,k}$ is essentially the conditional expectation of  the statistic $S^{(t,s)}_{j,k}$. The quantity $\frac{1}{\log m}$ here is the second-order correction term $\hat \delta$, which accounts for the contributions of all intersecting $w$-tuples of neighbors. The quantity $4^{s-i}$ is a fudge factor, correcting for some additional error terms, including the fact that $\Omega$ only obeys an approximate independence condition.

It is clear that the summands (S1) satisfy (P2), (P3). We must only properties (P4) and (P5), both of which hold vacuously for $\ell \neq t$. Thus, for the remainder of this section, we only consider $\ell = t$. Also, we assume throughout that $X$ is a non-empty set.

\begin{proposition}
The summands (S1) satisfy property (P5).
\end{proposition}
\begin{proof}
  Let $|X| = x$.   We must show that $\Phi^{(t,0)}_{\text{S1},t,k,j,X} \leq \Phi_{\text{S1},t,k,j,X}^{(t-1,s)} = m^{-100}$. We first estimate $S^{(t,0)}_{j,k}(X)$. For this, note that $C^{(t,0)}(Z) = 1$ for all $Z$ and $G^{(t)}$ is $\Delta$-semiconstrained. So
\begin{align*}
  S^{(t,0)}_{j,k}(X) &= \sum_{ \substack{Z \subseteq V \setminus X  \\ |Z| = k-j}} D_j^{(t)}(X \cup Z) = \tbinom{k}{j} D_k^{(t)}(X) \leq 2 \Delta^k 2^k (\log m)^{-g(k+x)}  \\
&\leq \Delta^k (\log m)^{-g(k+x)+1} \qquad \text{as $k \leq r \leq \log_2  \tfrac{\log m}{\log \log m} $}
\end{align*}

As $2^s = \Delta$, we have
$$
\Gamma^{(0)}_{j,k,x} = 2^{s (k-j)} \Delta^j (\log m)^{-g(k+x) + 4(k-j)} = \Delta^{k} (\log m)^{-g(k+x) + 4(k-j)}
$$

We thus compute  $\Phi^{(t,0)}_{\text{S1},t,k,j,X}$ as:
\begin{align*}
  \Phi^{(t,0)}_{\text{S1},t,k,j,X} &= 4^{s} \Bigl( \frac{S^{(t,0)}_{j,k}}{\Gamma^{(0)}_{j,k,x}} +  \frac{1}{\log m} \Bigr)^w \leq m^{10} \Bigl( (\log m)^{1 - 4 (k-j)} + \frac{1}{\log m} \Bigr)^w 
\end{align*}

As $j < k$, this is at most $m^{10} \bigl(  \frac{1}{\log^3 m} + \frac{1}{\log m}\bigr)^w$. As $w \geq \frac{1000 \log m}{\log \log m}$, we have $ \bigl( \frac{1}{\log^3 m} + \frac{1}{\log m} \bigr)^w = m^{-1000+o(1)}$; this shows $\Phi^{(t,0)}_{\text{(S1)},k,j,X} \leq m^{10} \times m^{-1000+o(1)}$; in particular, it is smaller than $m^{-100}$ for $m$ sufficiently large.
\end{proof}

We  next show property (P4) in two stages. First, we show that if the induction hypothesis (I) holds up to $(t,i)$, then this implies upper bounds on every term $S^{(t,i)}_{j,k}(X)$. We next use these upper bounds to compute $\bE[ \Phi^{(t,i)}_{\text{S1},t,k,j,X} \mid B^{(\ell, i-1)}]$, carrying the induction forward to sub-round $i+1$.

\begin{proposition}
  \label{Fubound}
If (I) holds up to $(t,i)$, then $S^{(t,i)}_{j,k}(X) \leq \Gamma^{(i)}_{j,k,|X|}$ for all sets $X \subseteq V$ and $j < k$.
  \end{proposition}
\begin{proof}
$\Phi^{(t,i)}_{\text{S1},t,k,j,X}$ is a potential summand so the induction condition (I) ensures that $\Phi^{(t,i)}_{\text{S1},t,k,j,X}  < 1$, i.e. $4^{s-i} \Bigl(  \frac{S^{(t,i)}_{j,k}}{\Gamma^{(i)}_{j,k,|X|}} + \frac{1}{\log m} \Bigr)^w < 1$. This in turn implies $\frac{S^{(t,i)}_{j,k}}{\Gamma^{(i)}_{j,k,|X|}} < 1$.
\end{proof}

\begin{proposition}
\label{hibound}
Let $X \subseteq V$  and $i \in \{1, \dots, s \}$.  When the entries $B^{\leq (t,i-1)}$ are fixed to satisfy (I) and $B^{(t,i)}$ is drawn from $\Omega$, we have 
$$
\bE[ (S^{(t,i)}_{j,k}(X))^{w} ] \leq 2 \Bigl( S^{(t,i-1)}_{j,k}(X) 2^{-(k-j)} + \Gamma^{(i)} _{j,k,|X|} (\log m)^{-2.98} \Bigr)^w
$$
\end{proposition}
\begin{proof}
  Let $|X| = x$.
  For every $Z \subseteq V \setminus X$ of size $|Z| = k-j$, let us define
  $$
  a_Z = C^{(t,i-1)}(Z) \times D_{j}^{(t)}(X \cup Z)
  $$

  Since $C^{(t,i)}(Z) = C^{(t,i-1)}(Z) \prod_{v \in Z} B^{(t,i)}(v)$, we thus compute $S^{(t,i)}_{j,k}(X)$ as
  \begin{align*}
    S^{(t,i)}_{j,k}(X) &= \sum_{\substack{Z \subseteq V \setminus X \\ |Z| = k-j}} C^{(t,i)}(Z) D_j(X \cup Z)  = \sum_{\substack{Z \subseteq V \setminus X \\ |Z| = k-j}} a_Z \prod_{v \in Z} B^{(t,i)}(v)
  \end{align*}

So we have:
  $$
  \bE[ ( S^{(t,i)}_{j,k}(X))^w ] = \bE \Bigl[ \Bigl( \sum_{\substack{Z \subseteq V \setminus X \\ |Z| = k-j}} a_Z \prod_{v \in Z} B^{(t,i)}(v) \Bigr)^w \Bigr]
  $$
  This is precisely the type of polynomial analyzed in Section~\ref{conc-polynomial}, with parameter $q = k-j$ and $\lambda = 1/2$. We will apply Theorem~\ref{Fapp-prop2}, wherein  the underlying random variables are $B^{(t,i)}$.  Since $k \leq r \leq O(\log m)$, we have $w q \leq O(\log m)$. The variables $B^{(t,i)}$ satisfy independence property (Q), which gives near-independence on tuples up to size $L \log m$. Therefore, when $L$ is a sufficiently large constant, the hypotheses of Theorem~\ref{Fapp-prop2} are satisfied for $\epsilon = \frac{1}{2 s}$ and $\lambda = \frac{1}{2}$, and so:
$$
\bE[ ( S^{(t,i)}_{j,k}(X))^w ] \leq \bigl(1 + \frac{1}{2 s} \bigr)  \Bigl(\sum_{\ell=0}^{q} \binom{w q}{\ell} \mu_{\ell} \Bigr)^{w} 
$$

As $1 + \frac{1}{2 s} \leq 2$ and $w q \leq O(\log m) \leq (\log m)^{1.01}$, we bound this somewhat crudely as:
\begin{equation}
  \label{Fttr2}
\bE[ ( S^{(t,i)}_{j,k}(X))^w ] \leq 2  \Bigl(\sum_{\ell=0}^{q} (\log m)^{1.01 \ell} \mu_{\ell} \Bigr)^{w} 
\end{equation}

We now turn to bounding the terms $\mu_{\ell}$. At $\ell = 0$ we have
\begin{equation}
  \label{mu0-e}
\mu_0 = \sum_{\substack{Z \subseteq V \setminus X \\ |Z| = q}} a_Z \lambda^q = \sum_{\substack{Z \subseteq V \setminus X \\ |Z| = q}} C^{(t,i-1)}(Z) \times D_{j}^{(t)}(X \cup Z) = 2^{-(k-j)} S^{(t,i-1)}_{j,k}(X)
  \end{equation}

To bound $\mu_{\ell}$ for $1 \leq \ell \leq k - j - 1$, consider some $W \subseteq V \setminus X$ with $|W| = \ell$. Then:
{\allowdisplaybreaks
\begin{align*}
  \sum_{\substack{Z \subseteq V \setminus X \\ |Z| = q \\ Z \supseteq W}} a_Z \lambda^{q - \ell} &= 2^{-(k-j-\ell)} \sum_{\substack{Z \subseteq V \setminus X \\ |Z| = q \\ Z \supseteq W}} C^{(t,i-1)}(Z) \times D_{j}^{(t)}(X \cup Z) \\
  &\leq 2^{-(k-j-\ell)} \sum_{\substack{U \subseteq V \setminus (X \cup W)  \\ |U| = q-\ell}} C^{(t,i-1)}(U) \times D_{j}^{(t)}(X \cup W \cup U) \qquad \text{letting $U = Z \setminus W$} \\
  &= 2^{-(k-j-\ell)} S_{j,k-\ell}^{(t,i-1)}(X \cup W) \leq 2^{-(k-j-\ell)} \Gamma^{(i-1)}_{j,k-\ell,x + \ell} \qquad \text{(by Proposition~\ref{Fubound})}
\end{align*}
}

We can simplify this as $\frac{ 2^{-(k-j-\ell)} \Gamma^{(i-1)}_{j,k-\ell,x + \ell}  }{\Gamma^{(i)}_{j,k,x}} = 2^{-\ell (s-i)} (\log m)^{-4 \ell}$, and thus we have shown:
\begin{equation}
  \label{mu1-e}
\mu_{\ell} \leq 2^{-\ell(s-i)} (\log m)^{-4 \ell} \Gamma^{(i)}_{j,k,x} \qquad \text{for $\ell = 1, \dots, k-j-1$}
\end{equation}

Finally, for $\ell = k-j$, we use the fact the $G^{(t)}$ is $\Delta$-semiconstrained to compute:
\begin{equation}
\label{mu2-e}
\mu_{k-j} =  \max_{\substack{Z \subseteq V \setminus X \\ |Z| = q}} C^{(t,i-1)}(Z) D_j^{(t)} (X \cup Z) \leq 2 \Delta (\log m)^{-g(k+x)} = 2 \Gamma_{j,k,x} (\log m)^{-4(k-j)} 
\end{equation}

Substituting the three bounds (\ref{mu0-e}), (\ref{mu1-e}), and (\ref{mu2-e}) into (\ref{Fttr2}) yields:
\begin{align*}
 \bE[ ( S^{(t,i)}_{j,k}(X))^w ] &\leq 2 \Bigl( S^{(t,i-1)}_{j,k}(X) 2^{-(k-j)} + 2  (\log m)^{1.01(k-j)} \Gamma_{j,k,x}  (\log m)^{-4(k-j)} \\
 &\qquad \qquad + \sum_{\ell=1}^{k-j-1} (\log m)^{1.01 \ell} \Gamma^{(i)}_{j,k,x} (\log m)^{-4 \ell} 2^{-\ell(s-i)} \Bigr)^w
 \end{align*}
 
 These summands decrease exponentially, and so the sum can be bounded by twice the summand at $\ell = 1$. Allowing another factor of $(\log m)^{0.01}$ to cover any constants, we get
\begin{align*}
 \bE[ ( S^{(t,i)}_{j,k}(X))^w ] \leq 2 \Bigl( S^{(t,i-1)}_{j,k}(X) 2^{-(k-j)} +  \Gamma^{(i)}_{j,k,x} (\log m)^{-2.98} \Bigr)^w
 \end{align*}
 which concludes the proof.
\end{proof}

\begin{proposition}
  \label{hibound2}
  Let $X \subseteq V$ and $i \in \{1, \dots, s \}$.  When the entries $B^{\leq (t,i-1)}$ are fixed to satisfy (I) and $B^{(t,i)}$ is drawn from $\Omega$,  we have 
  $$
  \bE[\Phi_{\text{S1},t,k,j,X}^{(t,i)} ] \leq 2 \times 4^{s-i} \Bigl( \frac{S_{j,k}^{(t,i-1)} (X)  }{2^{k-j} \Gamma^{(i)}_{j,k,|X|}} + (\log m)^{-2.99} + \frac{1}{\log m} \Bigr)^w
  $$
  \end{proposition}
  \begin{proof}
  Let $|X| =x$.   We have
  $$
  \bE[\Phi_{\text{S1},t,k,j,X}^{(t,i)} ] = 4^{s-i} \bE \Bigr[ \Bigl( \frac{S_{j,k}^{(t,i)} (X)  }{\Gamma^{(i)}_{j,k,x}} + \frac{1}{\log m} \Bigr) ^w \Bigr]
  $$

Let us consider the increasing concave-down function $f: \mathbb R_{+} \rightarrow \mathbb R_{+}$ defined by
$$
f(y) = \bigl( \frac{y^{1/w}}{\Gamma^{(i)}_{j,k,x}} + \frac{1}{\log m} \bigr)^w
$$
and let us also define $Q = (S^{(t,i)}_{j,k}(X))^w$. Proposition~\ref{hibound} shows that 
$$
\bE[Q] \leq 2 \alpha^w \qquad \qquad \text{where $\alpha = S^{(t,i-1)}_{j,k}(X) 2^{-(k-j)} + \Gamma_{j,k,x}^{(i)} (\log m)^{-2.98}$}
$$

We see that  $\bE[\Phi_{\text{S1},t,k,j,X}^{(t,i)} ] = 4^{s-i} \bE[ f(Q) ]$. By Jensen's inequality, we have
$$
\bE[f(Q)] \leq f(\bE[Q]) \leq f(2 \alpha^w) = \bigl( \frac{ 2^{1/w} \alpha}{\Gamma_{j,k,x}^{(i)}} + \frac{1}{\log m} \bigr)^w \leq 2 \bigl( \frac{ \alpha }{\Gamma_{j,k,x}^{(i)}} + \frac{1}{\log m} \bigr)^w
$$

Substituting the value of $\alpha$, we have
\begin{align*}
\frac{\alpha}{\Gamma_{j,k,x}^{(i)}} &= \frac{S^{(t,i-1)}_{j,k}(X) 2^{-(k-j)} + \Gamma_{j,k,x}^{(i)} (\log m)^{-2.98} }{\Gamma_{j,k,x}^{(i)}} = \frac{S^{(t,i-1)}_{j,k}(X)}{\Gamma_{j,k,x}^{(i-1)}} + (\log m)^{-2.98} 
\end{align*}

Thus, overall we have shown
\[
  \bE[\Phi_{\text{S1},t,k,j,X}^{(t,i)} ] = 4^{s-i} \bE[Q] \leq 2 \times 4^{s-i}  \bigl( \frac{S^{(t,i-1)}_{j,k}(X)}{\Gamma_{j,k,x}^{(i-1)}} + (\log m)^{-2.98}  + \frac{1}{\log m} \bigr). \qedhere
  \]
  \end{proof}

\begin{proposition}
The summand $\Phi_{\text{S1},t,k,j,X}$  satisfies Property (P4).
\end{proposition}
\begin{proof}
  Let $|X| = x$. Let us fix a value for $B^{(t, i-1)}$ satisfying (I). By Proposition~\ref{hibound2}, we calculate the ratio $\bE[\Phi_{\text{S1},t,k,j,X}^{(t,i)} ] /\Phi_{\text{S1},t,k,j,X}^{(t,i-1)}$ as:
$$
  \frac{\bE[ \Phi_{\text{S1},t,k,j,X}^{(t,i)} ]}{\Phi_{\text{S1},t,k,j,X}^{(t,i-1)}} \leq  \frac{ 2 \times 4^{s-i} \Bigl(  \frac{S^{(t,i-1)}_{j,k}(X)}{\Gamma^{(i-1)}_{j,k,x}}  + (\log m)^{-2.98} + \frac{1}{\log m} \Bigr)^w }{4^{s-i+1} \Bigl( \frac{S^{(t,i-1)}_{j,k}  }{\Gamma^{(i-1)}_{j,k,x}} + \frac{1}{\log m} \Bigr)^w } = \tfrac{1}{2} \Biggl( 1 + \frac{ (\log m)^{-2.98} }{\frac{S^{(t,i-1)}_{j,k}  }{\Gamma^{(i-1)}_{j,k,x}} + \frac{1}{\log m} } \Biggr)^w
  $$
  
  We bound this latter term as:
  $$
  \frac{ (\log m)^{-2.98} }{\frac{S^{(t,i-1)}_{j,k}  }{\Gamma^{(i-1)}_{j,k,x}} + \frac{1}{\log m} } \leq 
  \frac{ (\log m)^{-2.98} }{\frac{1}{\log m} }  = (\log m)^{-1.98}
  $$
and so we have shown that
  $$
  \frac{\bE[ \Phi_{\text{S1},t,k,j,X}^{(t,i)}  ]}{\Phi_{\text{S1},t,k,j,X}^{(t,i-1)}} \leq  \tfrac{1}{2} \bigl( 1 + (\log m)^{-1.98} \bigr)^w \leq \tfrac{1}{2} e^{ w (\log m)^{-1.99} }
  $$

Since $w \leq O(\log m)$, this shows that $\frac{\bE[ \Phi_{\text{S1},t,k,j,X}^{(t,i)} ]}{\Phi_{\text{S1},t,k,j,X}^{(t,i-1)}} \leq \tfrac{1}{2} (1+o(1)) \leq 1$ for $m$ sufficiently large.
\end{proof}

\begin{theorem}
  \label{Fmain-result-det}
  If (I) holds up to $(t,s)$, then every $X \subseteq V$ has 
  $$
  D_j^{(t+1)}(X) \leq D_j^{(t)}(X) + \Delta^j (\log m)^{-g(j+1+|X|) + 4.01}.
  $$
\end{theorem}
\begin{proof}
  Let $|X| = x$.  By Proposition~\ref{Fubound}, we have $M_{j,k}^{(t)}(X) \leq S_{j,k}^{(t,s)}(X) \leq \Gamma^{(s)}_{j,k,x} = \Gamma_{j,k,x}$.  Summing over $k = j+1$ to $q$, we get a total migration of
$$
\sum_{k=j+1}^q M^{(t)}_{j,k}(X) \leq \sum_{k=j+1}^q \Gamma_{j,k,x}
$$

In this sum, the ratio between the $k+1$ term and $k$ term is given by
$$
\frac{\Gamma_{j,k+1,x}}{\Gamma_{j,k,x}} = (\log m)^{-g(k+1+x) + g(k+x) + 4} = (\log m)^{-2^{k+x+2} + 4}
$$
For $x \geq 1$ and $k \geq 2$, this is at most $(\log m)^{-28} \ll 1$. So the summands decrease exponentially and we can bound it by twice the summand at $k = j+1$, namely
\[
\sum_{k>j} \Gamma_{j,k,x} \leq 2 \Gamma_{j,j+1,x} = 2 \Delta^j (\log m)^{-g(j+1+x) + 4} \leq  \Delta^j (\log m)^{-g(j+1+x) + 4.01} \qedhere
\]
\end{proof}

\section{Finishing the induction}
\label{Fcollapse-det-sec}
We now discuss how to enforce the edge-collapse, similar to the randomized algorithm, and use this to show the induction condition (I).

To begin, we  introduce a collection of potential summands enforcing edge collapse. Specifically, for every set $X \subseteq V$ and every triple of integers $j, \tau, t$ with $1 \leq j \leq r$ and $1 \leq \tau < t \leq T$, we have a potential summand $\Phi_{\text{S2}, X, j, \tau, t}$ to ensure that $D^{(k)}_j(X)$ is small for some prior round $k$.  
These calculations are routine derandomizations of the randomized algorithm, and are also quite similar to calculations of Kelsen \cite{kelsen}, so we defer the formal definition and analysis of $\Phi_{\text{S2}, X, j, \tau, t}$  to  Appendix~\ref{collapse-proofs}.  We summarize it as follows:
\begin{proposition}
  \label{Ftg-prop}
The potential summand $\Phi_{\text{S2}, X, j, \tau, t}$ satisfies properties (P2) -- (P5). Furthermore, if $\Phi_{\text{S2}, X,j, \tau, t}^{(t-1, s)} < 1$ then $D_j^{(k)}(X) \leq \frac{\Delta^j (\log m)^{1.01}}{\tau}$ for some $k \in \{t - \tau+1, \dots, t \}$.
  \end{proposition}

Using Proposition~\ref{Ftg-prop}, we can complete the induction argument for DET-REDUCE.
\begin{proposition}
  \label{Ftg-prop2}
  If condition (I) holds up to $(t-1, s)$, then for every set $X \subseteq V$, and all integers $j, \tau$ with $j \geq 1, 1 \leq \tau < t \leq T+1$ we have $D_j^{(t)}(X) \leq \frac{\Delta^j (\log m)^{1.01}}{\tau} + \tau \Delta^j (\log m)^{-g(j+x+1) + 4.01}$.
\end{proposition}
\begin{proof}
Condition (I) up to round $(t-1,s)$ ensures that $\Phi^{(t-1, s)}_{\text{S2},X, j, \tau, t} < 1$. Therefore, by Proposition~\ref{Ftg-prop}, we have $D_j^{(k)}(X) \leq \frac{\Delta^j (\log m)^{1.01}}{\tau}$ for some $k \in \{t- \tau+1, \dots, t \}$. Also, by Theorem~\ref{Fmain-result-det}, $D_j(X)$ increases by at most $\Delta^j (\log m)^{-g(j+1+|X|) + 4.01}$ in each such round. Thus,
\begin{align*}
  D_j^{(t)}(X) &\leq D_j^{(k)}(X) + (t-k) \Delta^j (\log m)^{-g(j+x+1) + 4.01} \leq \tfrac{\Delta^j (\log m)^{1.01}}{\tau} + \tau \Delta^j (\log m)^{-g(j+x+1) + 4.01}  \qedhere
 \end{align*}
\end{proof}

\begin{proposition}
Condition (P1) is satisfied.
\end{proposition}
\begin{proof}
There is a potential summand (S1) for each choice of integers $t,j,k$, and each $X \subseteq V$. There are $O(\log \log m)$ choices for $j,k$ and there are $T = (\log m)^{2^{r+2}}$ choices for $t$; observe that $(\log m)^{2^r} \leq n$ by our condition on $r$, so this is at most $O(m^4)$ choices. 

There appear to be $2^n$ choices for $X$, which would be exponential.  However, observe that these summands are only non-trivial if $X$ is a subset of an edge of the original input hypergraph $G$. There are $m$ edges and each edge has at most $2^r \leq O(\log m)$ subsets, so in total the number of summands is at most $O(m \log m)$.

In all, there are at most $O(m^{5} \log m)$ potential summands of type (S1). A similar argument applies to the potential summands of type (S2).
\end{proof}

\begin{proposition}
\label{Fbb1}
The induction condition (I) hold for all rounds $t = 1, \dots, T+1$
\end{proposition}
\begin{proof}
We prove by induction on $t$ that (I) holds up to $(t,0)$ for every $t$ in this range. For the base case $t = 1$, note that there are fewer than $m^{100}$ summands, and each $\sigma$ has $\Phi^{(0,s)}_{\sigma} \leq m^{-100}$. Therefore, $\Phi^{(0,s)} < 1$. By hypothesis, the input hypergraph $G^{(1)}$ is $\Delta$-constrained. Thus, by (P5), $\Phi^{(1,0)} \leq \Phi^{(0,s)}$.

Now suppose that (I) holds up to $(t-1,0)$. So  $G^{(1)}, \dots, G^{(t-1)}$ are $\Delta$-semiconstrained and $\Phi^{(t-1,0)} < 1$. We will show that (I) holds up to $(t-1, i)$ by induction on $i$ for $i = 0, \dots, s$. By the induction hypothesis, $\Phi^{(t-1,i-1)} < 1$.  So property (P4) ensures that $\bE[ \Phi^{(t,i)}_{\sigma}] \leq \Phi^{(t,i-1)}_{\sigma}$ for every summand $\sigma$, when $B^{(t,i)}$ is drawn from $\Omega$. Therefore, $\bE[ \Phi^{(t,i)}] \leq \Phi^{(t,i-1)}$. Since DET-REDUCE searches the support of $\Omega$ to minimize $\Phi^{(t,i)}$, this means that $\Phi^{(t,i)} \leq \Phi^{(t,i-1)} < 1$.

This shows that (I) holds up to $(t-1,s)$. We next claim that $G^{(t)}$ is $\Delta$-semiconstrained. Consider any integer $j \geq 1$ and set $X \subseteq V$ with $|X| = x \leq r - j$. We need to show that $D_j^{(t)}(X) \leq 2 \Delta^j (\log m)^{-g(j+x)}$. 

Let $\tau =  (\log m)^{g(j+x+1) - g(j+x) - 4.01}$. If $t \leq \tau$ then by Theorem~\ref{Fmain-result-det} the migration into $N_j(X)$ in each previous round is at most $\Delta^j (\log m)^{-g(x+j+1) + 4.01}$ and so
\begin{align*}
D_j^{(t)}(X) &\leq D_j^{(1)}(X) + (t-1) \Delta^j  (\log m)^{-4.01 + g(x+j+1)} \\
&\leq \Delta^j (\log m)^{-g(j+x)} + \Delta^j (\log m)^{g(j+x+1) - g(j+x) - 4.01}  (\log m)^{4.01 - g(x+j+1)} \\
&= 2 \Delta^j (\log m)^{-g(j+x)}
\end{align*}

If $t > \tau$, then Proposition~\ref{Ftg-prop2} gives
  \begin{align*}
\label{Frr1}
D_j^{(t)}(X) &\leq \frac{\Delta^j (\log m)^{1.01}}{\tau} + \tau \Delta^j (\log m)^{-g(j+x+1) + 4.01} \\
%&=  \Delta^j (\log m)^{5.02 - g(j+x+1) + g(j+x)} + \Delta^j (\log m)^{-g(j+x)} \\
&= \Delta^j (\log m)^{-g(j+x)} ( (\log m)^{5.02 - g(j+x+1) + 2 g(j+x)} + 1 )\\
&= \Delta^j (\log m)^{-g(j+x)} ( (\log m)^{-3.98} + 1) \leq 2 \Delta^j (\log m)^{-g(j+x)}
\end{align*}

As this holds for all such $X, j$, we have shown that $G^{(t)}$ is $\Delta$-semiconstrained. 

Finally, the property (P5) implies that $\Phi^{(t,0)} \leq \Phi^{(t-1,s)} < 1$. Thus, (I) holds up to $(t,0)$, completing the induction.
\end{proof}

\begin{theorem}
\label{Fbig-thm-det}
$G^{(T+1)}$ is $\tfrac{1}{2} \Delta$-constrained.
\end{theorem}
\begin{proof}
  
  Consider any integer $j \geq 1$ and set $X \subseteq V$ with $|X| = x \leq r - j$ and let
  $$
  \tau = \tfrac{1}{2} 2^{-j} (\log m)^{g(j+x+1) - g(j+x) - 4.01}
  $$

  We can easily check that $\tau \leq T$. By Proposition~\ref{Fbb1}, (I) holds up to $(T,s)$. Therefore, applying Proposition~\ref{Ftg-prop2} at round $t = T+ 1$ gives
    \begin{align*}
    D_j^{(T+1)}(X) &\leq \frac{\Delta^j (\log m)^{1.01}}{\tau} + \tau \Delta^j (\log m)^{-g(j+x+1) + 4.01} \\
&= 2^{j+1} \Delta^j (\log m)^{5.02 - g(j+x+1) + g(j+x)} + \tfrac{1}{2} (\tfrac{\Delta}{2})^j (\log m)^{-g(j+x)} \\
&= (\tfrac{\Delta}{2})^j (\log m)^{-g(j+x)} \Bigl( 2^{2j+2} (\log m)^{5.02 - g(j+x+1) + 2 g(j+x)} + \tfrac{1}{2} \Bigr)
  \end{align*}

  As $j \leq r \leq \log_2( \frac{\log n}{\log \log n} )$,  we have $2^{2 j + 2} \leq (\log m)^2$, and therefore
  \begin{align*}
    D_j^{(T)}(X) &\leq (\tfrac{\Delta}{2})^j (\log m)^{-g(j+x)} \Bigl( (\log m)^{7.02 - g(j+x+1) + 2 g(j+x)} + \tfrac{1}{2} \Bigr) \\
    &= (\tfrac{\Delta}{2})^j (\log m)^{-g(j+x)} \Bigl( (\log m)^{-1.98} + \tfrac{1}{2} \Bigr) \leq (\tfrac{\Delta}{2})^j (\log m)^{-g(j+x)}
  \end{align*}

  Since this holds for arbitrary $X, j$ it implies that $G^{(T+1)}$ is $\tfrac{\Delta}{2}$-constrained.
\end{proof}

{
\renewcommand{\thetheorem}{\ref{Fdet-thm-overall}}
\begin{theorem}
There is a deterministic algorithm, running in $(\log n)^{2^{r+3} + O(1)}$ time and $\poly(m,n)$ processors, to produce an MIS of a rank-$r$ hypergraph.
\end{theorem}
}
\begin{proof}
  Use Algorithm~\ref{algo:find-mis}, replacing REDUCE with DET-REDUCE.
\end{proof}

\section{Sparse hypergraphs}
\label{Fsparse-sec}
Bercea et al. \cite{bercea} introduced an MIS algorithm named SBL for hypergraphs with relatively few edges. We summarize this here:
\begin{algorithm}[H]
\centering
\begin{algorithmic}[1]
\State{Initialize $I = \emptyset$}
\While{$G^{\text{res}}_I \neq \emptyset$}
\State{Mark each vertex independently with probability $p$}
\State{If an edge with more than $r$ vertices is fully marked, unmark one arbitrary vertex.}
\State{Let $X$ denote the vertices which remain marked.}
\State{Find an MIS of the hypergraph $G^{\text{res}}_I[X]$, and add it to $I$.}
\EndWhile
\end{algorithmic}
\caption{The SBL algorithm}
\label{sbl-alg}
\end{algorithm}

To analyze this algorithm, let us define a vertex $v$ to be \emph{free} for an independent set $I$ if $v \notin I$ and $\{v \}$ is not an edge of $G^{\text{res}}_I$. An independent set $I$ is maximal iff there are no free vertices for $I$. The key insight of \cite{bercea} is that since line (6) produces an MIS, the vertices in $X$ become non-free. Thus for $p$ chosen appropriately, the number of free vertices drops by a $m^{-1/r}$ factor in each round.  We now derandomize this step of choosing the vertex set $X \subseteq V$.
\begin{proposition}
  \label{Ff1}
Suppose that hypergraph $G$ has $m \geq 1$ edges and $n'$ free vertices for an independent set $I$. For any integer $r \geq 1$, there is a deterministic algorithm in $\polylog(m,n)$ time and $\poly(m,n)$ processors to produce a set $X \subseteq V$ of free vertices, such that $G^{\text{res}}_I[X]$ has rank $r$ and $|X| \geq \Omega( n' m^{-1/r} )$.
\end{proposition}
\begin{proof}
We assume that $n' > 0$ as otherwise this is trivial.  Let $p = (2 m)^{-1/r}$ and $G' = G^{\text{res}}_I$, and consider the following random process: we put each free vertex into $Y$ independently with probability $p$; if $e \subseteq Y$ for any edge $e$ of $G'$ with $|e| > r$, then we remove one arbitrary vertex of $e$ from $Y$. We let $X$ denote the resulting vertex set. This process ensures that $G'[X] = G^{\text{res}}_I[X]$ has rank $r$ and $\bE[|X|] \geq n' p - m p^{r+1} \geq  n' p - p/2 \geq n' p /2$.
  
  We derandomize this process using a general methodology of Sivakumar \cite{sivakumar}. Observe that, for any edge $e$, the event $e \subseteq Y$ can be represented via a ``log-space statistical test''; specifically, as we process the vertices in order, we check whether every vertex $v \in e$ is added to $Y$. There are a polynomial number of such tests (one for each edge), so one can efficiently construct a probability space $\Omega$ with $\poly(m,n)$ support size fooling them all to error $\epsilon = \frac{p}{4 m} \geq \frac{1}{\poly(m,n)}$.

  When the selection vector $Y$ is drawn from $\Omega$, we have $\bE[|X|] \geq np - m p^{r+1} - m \epsilon \geq n' p - p /2 - p/4 \geq n'p/4 $. In particular, there is at least one value in the support of $\Omega$ such that $|X| \geq n' p/4 \geq \Omega(n' m^{-1/r})$. Since $X$ has polynomial support, we can search it efficiently in $\polylog(m,n)$ time and $\poly(m,n)$ processors.
\end{proof}

Thus, the following DSBL algorithm is a derandomized version of SBL:
\begin{algorithm}[H]
\centering
\caption{The DSBL algorithm}
\label{dsbl-alg}
\begin{algorithmic}[1]
\State{Initialize $I = \emptyset$}
\For{$t = 1, 2, \dots, $ until $G^{\text{res}}_I$ is empty}
\State{Apply Proposition~\ref{Ff1} to obtain vertex set $X_t$.}
\State{Find an MIS of the hypergraph $G^{\text{res}}_I[X_t]$, and add it to $I$.}
\EndFor
\end{algorithmic}
\end{algorithm}

{
\renewcommand{\thetheorem}{\ref{Fth-sbl1}}
\begin{theorem}
There is an algorithm time for hypergraph MIS using  $\exp(O( \frac{\log(m n)}{\log \log(m n)})) $ expected time and $O(n + m \log n)$ processors. There is a deterministic algorithm for hypergraph MIS in $\exp(O( \frac{\log (mn)}{\log \log (m n)}))$ time and $\poly(m,n)$ processors.
\end{theorem}
}
\begin{proof}
We assume $m \geq 1$ as otherwise this is trivial. The deterministic (respectively randomized) algorithm is to apply DSBL (respectively SBL) with $r = \lceil \log_2  \tfrac{\log(m n)}{(\log \log (m n))^2}  \rceil$. We only analyze the deterministic case, as the randomized algorithm is nearly idential.

For each round $t$, let $n'_t$ denote the number of free vertices for $I$. Note that $|X_t| \geq \Omega( n'_t m^{-1/r})$ and so $n'_{t+1} \leq n'_t - |X_t|   \leq n'_t (1 - \Omega(m^{-1/r}))$. This implies that, for some $t = O(m^{1/r} \log n)$ we have $n'_t = 0$; then the independent set $I$ is maximal and DSBL algorithm terminates.  With this parameter $r$ we have $m^{1/r} = \exp( O( \frac{\log (m n)}{\log \log (m n)}))$.

By Proposition~\ref{Ff1}, it requires $\polylog(m, n)$ time to produce the set $X_t$.  Each hypergraph $G_t[X_t]$ has rank $r$ by construction, so we can find an MIS of $G_t[X_t]$ using Theorem~\ref{Fdet-thm-overall} in time $(\log n)^{2^{r+3}+O(1)} = \exp( O(\frac{\log (m n)}{\log \log (m n)}))$.
\end{proof}

\section{Conclusion}
\label{conclude-sec}
We have examined two related subjects: the algorithmic problem of hypergraph MIS, and the technical tool of derandomized concentration bounds for polynomials applied to independent random variables. Let us provide an overview of where these now stand and future directions for them.

At this stage, we have NC algorithms for hypergraph MIS of fixed rank $r$. We suspect that an efficient general MIS algorithm should exist; as far as we are aware, the randomized algorithm of Beame \& Luby \cite{beame-luby} is likely to already run in $\polylog(n)$ rounds for general hypergraphs. We note that our proof strategy for this algorithm, based on globally bounding the degree, is relatively weak. A similar algorithm and proof strategy was used for graph MIS by Blelloch, Fineman, \& Shun \cite{greedy2}, showing a convergence time of $O(\log^2 n)$ rounds. Fischer \& Noever \cite{fischer} later provided a more sophisticated analysis showing that this algorithm in fact runs in $O(\log n)$ rounds, matching Luby's MIS algorithm (which is based on tracking global edge count). Instead of simply tracking degree, they analyze long dependency chains through the vertices. We do not know how to extend such analysis, which is already complex for ordinary graphs, to general hypergraphs.

Even if we cannot obtain algorithms for general hypergraphs, we still see much room for improvement for fixed-rank hypergraphs. A runtime of $c_r (\log n)^{\poly(r)}$ (where $c_r$ could be an arbitrary parameter) would already be a significant advance. We are not aware of any algorithm with a runtime even of $(\log n)^{O(1.99^r)}$.

The main technical tool of this algorithm is (derandomization of) concentration for polynomials. This is an important subject in probability theory on its own, and is likely to have applications to other algorithms. Our derandomization method is based on conditional expectations with appropriate potential functions. Unfortunately, this is quite messy and is also somewhat specialized to the parameter ranges needed for the hypergraph MIS algorithm. 

By contrast, for the \emph{randomized} analysis of concentration bounds for polynomials, Schudy \& Sviridenko \cite{schudy-sviridenko} have clean bounds in terms of easy-to-compute statistics. These bounds apply to a much larger class of polynomials, including settings in which the underlying input variables are not Bernoulli. We see in Proposition~\ref{Fmigrate-prop1} how to apply these bounds to specific settings by a few simple computations.

The derandomization of sums of independent random variables also has a rich, robust theory behind it. For example, Alon \& Srinivasan \cite{alon-srin} give simple-to-compute potential functions which can be used for conditional expectations in probability spaces with almost-independence. Also, since sums of variables can usually be computed in logarithmic space, Sivakumar's method \cite{sivakumar} fools them in a nearly black-box way (we have already seen an example of this, in Proposition~\ref{Ff1}). In general, randomized processes based on such sums can typically be derandomized using well-understood and high-level techniques.

Our hope is that there may be some way to extract the main ideas from our analysis of derandomized polynomial bounds and package them in a way which is as general and clean.

\section{Acknowledgments}
Thanks to Aravind Srinivasan and Ioana Bercea for helpful discussions. Thanks to anonymous conference and journal reviewers for many suggestions and corrections.

\appendix 

\section{Analysis of edge collapse}
\label{collapse-proofs}
In this section, we consider some set $X \subseteq V$ with $|X| = x > 0$, and analyze how $X$ collapses in both the randomized and deterministic algorithms. The analysis for the randomized algorithm is very similar to an argument given by Beame \& Luby \cite{beame-luby}, and so we only provide a sketch here.

\begin{repproposition}{Fcollapse-prop}
If $G^{(t)}$ is $\Delta$-semiconstrained, then $$
\Pr( \text{$X$ collapses at round $t$}) \geq \tfrac{1}{4} \sum_{k=1}^r \Delta^{-k} D^{(t)}_k(X).
$$
\end{repproposition}
\begin{proof}
  Let us fix round $t$, and we omit the superscript $(t)$ from the notations for readability. We begin by observing the inequality
  \begin{equation}
    \label{gt1}
  [ \text{$X$ collapses} ] \geq \sum_{Y \in N(X)} C(Y) \Bigl( 1 - \sum_{\substack{e \in G\\ e \cap Y \neq \emptyset}} C(e \setminus Y) - \sum_{\substack{Y' \in N(X) \\ Y \neq Y'}} C(Y' \setminus Y) \Bigr)
  \end{equation}

To show (\ref{gt1}), observe that, if $z$ is the total number of sets $Y \in N(X)$ with $C(Y) = 1$, then
$$
\sum_{Y \in N(X)} C(Y) \bigl( 1 - \sum_{\substack{Y' \in N(X) \\ Y \neq Y'}} C(Y' \setminus Y) \bigr) = z - \binom{z}{2}
$$

Consequently, the RHS of (\ref{gt1}) is only positive if $z = 1$, i.e. there is exactly one set $Y \in N(X)$ with $C(Y) = 1$. If $C(e \setminus Y) = 0$ for all edges $e$ intersecting $Y$, then $Y$ is added to the independent set, causing $X$ to collapse. Taking expectations of (\ref{gt1}), the probability that $X$ collapses is at least
$$
 \sum_{Y \in N(X)} \negthickspace \negthickspace \Pr(C(Y) = 1) \Bigl( 1 - \sum_{\substack{e \in G\\ e \cap Y \neq \emptyset}} \negthickspace \bE[ C(e \setminus Y) \mid C(Y) = 1]  -  \sum_{\substack{Y' \in N(X) \\ Y \neq Y'}} \negthickspace \bE[C(Y' \setminus Y) \mid C(Y) = 1] \Bigr)
    $$

Using the fact that $G$ is $\Delta$-semiconstrained, one can show that, for any fixed $Y \in N(X)$, the conditional expectations $\bE[\sum_{\substack{e \in G\\ e \cap Y \neq \emptyset}} C(e \setminus Y) \mid C(Y) = 1]$ and $\bE[\sum_{\substack{Y' \in N(X) \\ Y \neq Y'}} C(Y' \setminus Y) \mid C(Y) = 1]$ are both at most $0.01$.  (See Proposition~\ref{Fcollapse-prop-det} for further details). Therefore, 
\begin{equation}
   \label{gt3}
          \Pr(\text{$X$ collapses}) \geq \sum_{Y \in N(X)} \Pr(C(Y) = 1) ( 1 - 0.01 - 0.01 ) \geq \tfrac{1}{4} \sum_{Y \in N(X)} \bE[C(Y)]
\end{equation}

Each $Y \in N_k(X)$ has $\bE[C(Y)] = \Delta^{-k}$ so the RHS of (\ref{gt3}) is equal to $\sum_{k=1}^r \Delta^{-k} D_k(X)$.
\end{proof}

Proposition~\ref{Fcollapse-prop} is not quite enough for the deterministic algorithm, which requires showing that the probability bound holds for an approximate-independence probability space and it can be witnessed by an easy-to-compute pessimistic estimator.

\begin{proposition}
\label{q2-prop}
  Let $Y \subseteq V$ be any set of size $y \leq L \log m$. When $B^{(t,1)}, \dots, B^{(t,s)}$ are drawn independently from $\Omega$, we have $\tfrac{1}{2} \Delta^{-y} \leq \Pr( C^{(t)}(Y) = 1 ) \leq 2 \Delta^{-y}$.
\end{proposition}
\begin{proof}
  For each $i = 1, \dots, s$, property (Q) gives $(1-\frac{1}{2s}) 2^{-y} \leq \Pr(\bigwedge_{v \in Y} B^{(t,i)}(v)=1) \leq (1+\frac{1}{2 s}) 2^{-y}$.   Therefore, overall we have $((1-\frac{1}{2 s}) 2^{-y})^s \leq \Pr( C^{(t)}(Y) = 1)  \leq ((1+\frac{1}{2 s}) 2^{-y})^s$.   Now note that $2^s = \Delta$, and that $ (1 + \frac{1}{2 s})^s \leq 2$ and $(1 - \frac{1}{2 s})^s \geq \tfrac{1}{2}$.
\end{proof}

We now introduce a function $H^{(t)}$, which serves as a pessimistic estimator for the event that $X$ fails to collapse at round $t$. This quantity $H^{(t)}$ is a function of the bits $B^{\leq(t,s)}$.  We also define a related function $H^{(t,i)}$, which is the expectation of $H^{(t)}$, if the bits $B^{\leq(t,i)}$ are fixed and the bits $B^{(t,i+1)}, \dots, B^{(t,s)}$ are drawn independently from $\Omega$. We emphasize here that $H^{(t,i)}$ is completely determined by the function $H^{(t)}$ and that $H^{(t,s)} = H^{(t)}$.

\begin{proposition}
  \label{Fcollapse-prop-det}
  For each round $t \geq 1$ we can define the quantity $H^{(t)}$ to have the following properties:
  \begin{enumerate}
\item[(A1)] $H^{(t)}$ is an non-negative integer.
\item[(A2)] If $H^{(t)} = 0$, then $X$ collapses at round $t$.
\item[(A3)] If $G^{(t)}$ is $\Delta$-semiconstrained, then $H^{(t,0)}  \leq 1 - \tfrac{1}{4} \sum_{k=1}^r \Delta^{-k} D^{(t)}_k(X)$.
\item[(A4)] The quantities $H^{(t,i)}$ can be computed using $\poly(m,n)$ processors and $\polylog(m,n)$ time.
\end{enumerate}
\end{proposition}
\begin{proof}
For notational clarity, we omit the superscript $(t)$ throughout the remainder of this proof. We define $H$ by
$$
H = 1 - \sum_{Y \in N(X)} C(Y) \Bigl( 1 - \sum_{\substack{e \in G\\ e \cap Y \neq \emptyset}} C(e \setminus Y) - \sum_{\substack{Y' \in N(X) \\ Y \neq Y'}} C(Y' \setminus Y) \Bigr)
$$

To show (A1), observe that, if $z$ is the total number of sets $Y \in N(X)$ with $C(Y) = 1$, then
$$
\sum_{Y \in N(X)} C(Y) \bigl( 1 - \sum_{\substack{Y' \in N(X) \\ Y \neq Y'}} C(Y' \setminus Y) \bigr) = z - \binom{z}{2}
$$

Consequently, $H$ is a non-negative integer, and $H = 0$ only if $z = 1$, i.e. there is exactly one set $Y \in N(X)$ with $C(Y) = 1$. If $C(e \setminus Y) = 0$ for all edges $e$ intersecting $Y$, then $Y$ is added to the independent set, causing $X$ to collapse. This implies (A2). 

To show (A4), we can calculate $\bE[B^{(t,i)}(Y)]$ for any set $Y \subseteq V$ by enumerating over all possible values of $B^{(t,i)}$ in the support of $\Omega$. Since $\Omega$ has a support of size $\poly(m,n)$, this can be done in $\polylog(m,n)$ time and $\poly(m,n)$ processors. The conditional expectation, when $B^{(\leq t,i)}$ is fixed and $B^{(t,i+1)}, \dots, B^{(t,s)}$ are independently drawn from $\Omega$, is simply the $(s-i)$ power of this. 

Finally, to show (A3), we use Proposition~\ref{q2-prop} to estimate:
\begin{equation}
  \label{hheqn}
\bE[H \mid B^{\leq(t-1,s)}] \leq 1 - \sum_{Y \in N(X)} \tfrac{1}{2} \Delta^{-|Y|} \Bigl( 1 - \sum_{\substack{e \in G \\  e \cap Y \neq \emptyset}} 2 \Delta^{-|e \setminus Y|} - \sum_{\substack{Y \neq Y'\\ Y' \in N(X)}} 2 \Delta^{-|Y' \setminus Y|} \Bigr)
\end{equation}

Let us now fix $Y \in N_k(X)$ for some integer $k \geq 1$ and estimate the sums in (\ref{hheqn}) over edges $e \in G$ and over $Y' \in N(X)$. For the first quantity, set $U = e \cap Y \neq \emptyset$ and so we have:
\begin{align*}
  \sum_{\substack{e \in G \\  e \cap Y \neq \emptyset}} \negthickspace \Delta^{-|e \setminus Y|} &=  \sum_{\substack{U \subseteq Y \\ U \neq \emptyset}} \sum_{e \in G: e \cap Y = U} \Delta^{-|e \setminus U|} \leq  \sum_{\substack{U \subseteq Y \\ U \neq \emptyset}} \sum_{j = 0}^r D_j(U) \Delta^{-j}
  \end{align*}

Note that $Y \in N(X)$ so that $X \cup Y$ is an edge in $G$. Since $G$ is a residual hypergraph, it does not have nested edges and therefore any $U \subseteq Y$ cannot be an edge, i.e. $D_0(U) = 0$. Using this fact and the fact that that $G$ is $\Delta$-semiconstrained, we get:
\begin{align*}
\sum_{ \substack{U \subseteq Y \\ U \neq \emptyset}} \sum_{j = 0}^r D_j(U) \Delta^{-j} \leq \sum_{\substack{ U \subseteq Y \\ U \neq \emptyset}} \sum_{j = 1}^r 2 \Delta^j (\log m)^{-g(j+|U|)} \Delta^{-j} = 2 \sum_{u=1}^{k}  \sum_{j = 1}^r \tbinom{k}{u} (\log m)^{-g(j+u)}
\end{align*}

The summand here decreases exponentially in both $j$ and $u$. Consequently, the overall sum is bounded by a constant times the summand at $j = u = 1$, namely
\begin{align*}
\sum_{\substack{e \in G \\  e \cap Y \neq \emptyset}} \Delta^{-|e \setminus Y|} \leq O \bigl( k (\log m)^{-g(2)} \bigr)  \leq O( \frac{\log n}{\log \log n} (\log m)^{-7} ) \leq 0.01
\end{align*}

Let us next estimate the sum over $Y' \in N(X)$ in (\ref{hheqn}). Since $G$ is a residual hypergraph and $Y \neq Y'$, we cannot have $Y \subseteq Y'$. Let us therefore set $U = Y \cap Y' \subsetneq Y$, and we obtain:
{\allowdisplaybreaks
\begin{align*}
  \sum_{\substack{ Y' \in N (X)  \\ Y' \neq Y}} \Delta^{-|Y' \setminus Y|} &= \sum_{U \subsetneq Y} \sum_{j=1}^r \sum_{\substack{ Y' \in N_j (X)  \\ Y' \cap Y = U}} \Delta^{-j + |U|}  \leq \sum_{U \subsetneq Y} \Delta^{-j + |U|} D_{j-|U|} (X \cup U) \\
  &\leq \sum_{j=1}^r \sum_{U \subsetneq Y} 2 \Delta^{-j + |U|} \Delta^{j-|U|} (\log m)^{-g(k+x)} \qquad \text{as $G$ is $\Delta$-semiconstrained} \\
  &=2 \sum_{j=1}^r (\log m)^{-g(j+x)} \sum_{U \subsetneq Y} 1 =  (2^k - 1) \times 2 \sum_{j=1}^r (\log m)^{-g(j+x)} 
\end{align*}
}

We observe that $2^k \leq 2^r \leq \frac{\log m}{\log \log m}$ and thus these summands decrease exponentially, so this is at most $O( (\log m) \times (\log m)^{-g(2)}) = O( (\log m)^{-7}) \leq 0.01$.
  
Thus, we have shown that both of the sums over $e \in G$ and $Y' \in N(X)$ are bounded by $0.01$. Substituting into (\ref{hheqn}), we have shown:
\[
\bE[H \mid B^{\leq(t-1,s)}] \leq 1 - \sum_{Y \in N(X)} \tfrac{1}{2} \Delta^{-|Y|} \bigl( 1 - 2 \times 0.01 - 2 \times 0.01 \bigr) \leq 1 - \tfrac{1}{4}  \sum_{k=1}^r D_k^{(t)}(X) \Delta^{-k} \qedhere
\]
\end{proof}

We are now ready to define the potential summands $\Phi_{\text{S2}, X, j, \tau, t}$ as:
$$
\Phi_{\text{S2}, X, j, \tau, t}^{(\ell, i)} =
\begin{cases}
(1 - \frac{(\log m)^{1.01}}{4 \tau})^{\tau} & \text{if $\ell < t - \tau$} \\
 H^{(\ell, i)}   R_{\ell} (1 - \frac{(\log m)^{1.01}}{4 \tau})^{t-\ell}  & \text{if $t - \tau  \leq \ell < t$} \\
0 & \text{if $\ell \geq t$}
\end{cases}
$$
where, for each $b \in \{t - \tau+1, \dots, t \}$, we define the indicator variable $R_b$ by:
$$
R_b = \Bigl[ \bigwedge_{k=t - \tau+1}^b D_j^{(k)}(X) > \frac{\Delta^j (\log m)^{1.01}}{\tau} \Bigr]
$$ 

\begin{repproposition}{Ftg-prop}  
For all integers $j, t, \tau$ with $1 \leq j \leq r$ and $1 \leq \tau < t \leq T$, the potential summand $\Phi_{\text{S2}, X, j, \tau, t}$ satisfies properties (P2) -- (P5). Furthermore, if $\Phi_{\text{S2}, X,j, \tau, t}^{(t-1,s)} < 1$, then $D_j^{(k)}(X) \leq \frac{\Delta^j (\log m)^{1.01}}{\tau}$   for some $k \in \{t - \tau+1, \dots, t \}$.
\end{repproposition}
\begin{proof}
To simplify the notation,  we fix $j, \tau, t$ and we write  $\phi^{(\ell, i)}$ as shorthand for $\Phi^{(\ell,i)}_{\text{S2}, X, j, \tau, t}$ throughout.

Property (P2) follows immediately from (A4). To show property (P3), we compute
$$
\phi^{(0,s)} = (1 - \frac{(\log m)^{1.01}}{4 \tau})^{\tau} \leq e^{-\frac{(\log m)^{1.01}}{4}} \leq m^{-100}
$$

For property (P4), note that if $\ell < t - \tau$ or $\ell \geq t$ then $\phi^{(\ell,i+1)} = \phi^{(\ell,i)}$ and (P4) holds vacuously. If $t - \tau \leq \ell < t$, then property (P4) holds immediately by the law of iterated expectations.

Property (P5) holds trivially for $\ell \geq t $ or $\ell < t - \tau$. We will prove it only for the case $t - \tau < \ell < t$; the case $\ell = t - \tau$ is nearly identical and we omit it here. 

First, note that $R_{\ell} = 0$ then $\phi^{(\ell,0)} = 0$, so (P5) holds trivially. So assume that $R_{\ell} = 1$, i.e. $D_j^{(\ell)} > \frac{\Delta^j (\log m)^{1.01}}{\tau}$ and $R_{\ell-1} = 1$, in which case we have
\begin{align*}
  \frac{ \phi^{(\ell, 0)} }{\phi^{(\ell-1, s)}} &= \frac{ H^{(\ell,0)} (1 - \frac{(\log m)^{1.01}}{4 \tau})^{t - \ell} }{   H^{(\ell-1,s)} (1 - \frac{ (\log m)^{1.01}}{4 \tau})^{t - \ell + 1} } =  \frac{H^{(\ell, 0)}}{H^{(\ell-1,s)} (1 - \frac{ (\log m)^{1.01}}{4 \tau})}
\end{align*}

By property (A1), $H^{(\ell-1,s)} = H^{(\ell - 1)}$ is a non-negative integer. Furthermore, $D_j^{(\ell)}(X) > \frac{\Delta^j (\log m)^{1.01}}{\tau}$ so the set $X$ has evidently not collapsed at round $\ell-1$. Therefore, by (A2), we must have $H^{(\ell-1,s)} \geq 1$. By property (A3), we have $H^{(\ell, 0)} \leq 1-\tfrac{1}{4} \Delta^{-j} D_j^{(\ell)} (X)$; note here that we are assuming that $G^{(\ell)}$ is $\Delta$-semiconstrained as we are aiming to show property (P5). With our bound on $D_j^{(\ell)}(X)$, this implies that $H^{(\ell, 0)} \leq 1-\frac{(\log m)^{1.01}}{4 \tau}$. Therefore, we get
$$
  \frac{ \phi^{(\ell, 0)} }{\phi^{(\ell-1, s)}} \leq  \frac{H^{(\ell, 0)}}{H^{(\ell-1,s)} (1 - \frac{ (\log m)^{1.01}}{4 \tau})}
  \leq \frac{1 - (\frac{ (\log m)^{1.01}}{4 \tau})}{1 \times  (1 - \frac{ (\log m)^{1.01}}{4 \tau})} = 1
  $$
  which shows (P5).
  
  To show the second result of the proposition, note that $\phi^{(t-1,s)} = R_{t-1} H^{(t-1,s)}$. Again, $H^{(t-1,s)}= H^{(t-1)}$ is an non-negative integer. So the only way that we can have $\phi^{(t-1,s)} < 1$ if $R_{t-1} = 0$ or $H^{(t-1)} = 0$. In the former case, we clearly have $D_j^{(k)}(X) \leq \frac{\Delta^j (\log m)^{1.01}}{\tau}$ for some $k \in \{t - \tau+1, \dots, t-1 \}$. In the latter case, $X$ collapses at round $t-1$ and hence $D_j^{(t)}(X) = 0$.
\end{proof}
\end{document}